\documentclass[11pt,letterpaper]{article}
\usepackage{amsthm,color,latexsym,graphicx,url}
\usepackage[margin=1in]{geometry}
\usepackage{comment}
\usepackage[font=small,labelfont=bf]{caption}
\usepackage[labelformat=simple]{subcaption}

\usepackage{libertine}\usepackage[libertine]{newtxmath}
\usepackage[scaled=0.96]{zi4}
\usepackage[utf8]{inputenc}
\usepackage{fix-cm}

\usepackage{placeins}

\usepackage{lineno}

\graphicspath{{../}}

\newcommand{\vb}[1]{\texttt{#1}}

\title{You Can't Solve These Super Mario Bros.\ Levels:\texorpdfstring{\\}{ }Undecidable Mario Games}
\author{
  MIT Hardness Group%
    \thanks{Artificial first author to highlight that the other authors (in
      alphabetical order) worked as an equal group. Please include all
      authors (including this one) in your bibliography, and refer to the
      authors as “MIT Hardness Group” (without “et al.”).}
\and
  Erik D. Demaine%
    \thanks{MIT Computer Science and Artificial Intelligence Laboratory,
      32 Vassar St., Cambridge, MA 02139, USA, \protect\url{{edemaine,hhall314,hayastl,ruizr,nvenkat}@mit.edu}}
\and
  Lilly Hall\footnotemark[2]
\andlinebreak
  Hayashi Layers\footnotemark[2]
\and
  Ricardo Ruiz\footnotemark[2]
\and
  Naveen Venkat\footnotemark[2]
}
\date{}

\newif\ifabstract
\abstracttrue
\newif\iffull
\ifabstract \fullfalse \else \fulltrue \fi

\usepackage{hyperref}
\hypersetup{breaklinks,bookmarksnumbered,bookmarksopen,bookmarksopenlevel=2}
{\makeatletter \hypersetup{pdftitle={\@title}}}

{\makeatletter
 \gdef\xxxmark{%
   \expandafter\ifx\csname @mpargs\endcsname\relax 
     \expandafter\ifx\csname @captype\endcsname\relax 
       \marginpar{xxx}
     \else
       xxx 
     \fi
   \else
     xxx 
   \fi}
 \gdef\xxx{\@ifnextchar[\xxx@lab\xxx@nolab}
 \long\gdef\xxx@lab[#1]#2{\textbf{[\xxxmark #2 ---{\sc #1}]}}
 \long\gdef\xxx@nolab#1{\textbf{[\xxxmark #1]}}
}

{\makeatletter \gdef\fps@figure{!htbp}}

\let\realbfseries=\bfseries
\def\bfseries{\realbfseries\boldmath}

\def\andlinebreak{\end{tabular}\linebreak\begin{tabular}[t]{c}}

\newtheorem{theorem}{Theorem}[section]

\newtheorem{corollary}[theorem]{Corollary}
\theoremstyle{definition}

\let\epsilon=\varepsilon
\def\defn#1{\textbf{\textit{\boldmath #1}}}
\def\ICON#1{\hspace{0.1em}\raisebox{-0.6ex}{\includegraphics[height=\baselineskip]{figures/icons/#1}}\hspace{0.1em}\nolinebreak}

\newcommand{\figref}[1]{Figure~\ref{#1}}

\begin{document}
\maketitle

\begin{abstract}
  We prove RE-completeness (and thus undecidability) of several
  2D games in the Super Mario Bros.\ platform video game series:
  the New Super Mario Bros.\ series (original, Wii, U, and 2), and
  both Super Mario Maker games in all five game styles
  (Super Mario Bros.\ 1 and 3, Super Mario World, New Super Mario Bros.\ U,
  and Super Mario 3D World).
  These results hold even when we restrict to constant-size levels
  and screens, but they do require generalizing to allow
  arbitrarily many enemies at each location and onscreen,
  as well as allowing for exponentially large (or no) timer.
  Our New Super Mario Bros.\ constructions fit within one standard screen size.
  In our Super Mario Maker reductions, we work within the standard screen size
  and use the property that the game engine remembers offscreen objects
  that are global because they are supported by ``global ground''.
  To prove these Mario results, we build a new theory of counter gadgets
  in the motion-planning-through-gadgets framework, and provide a suite
  of simple gadgets for which reachability is RE-complete.
\end{abstract}

\section{Introduction}

In 2014, Hamilton \cite{braid} proved that solving a (bounded-size) level
in the 2008 video game \emph{Braid} is \defn{RE-complete}
(the same difficulty as the halting problem, and thus undecidable),
even without its famous time-travel mechanic.
The reduction is from \defn{counter machines} \cite{counter,Minsky-book,smallCounters}:
a finite-state machine with a constant number of natural-number counters
equipped with increment, decrement, and jump-if-zero instructions.
The central idea was to represent the value of each counter in a Braid level
by the number of enemies occupying a particular location in the level,
exploiting that this number can be arbitrarily large even in a
bounded-size level.
The same paper conjectured that many other video games were amenable
to this approach.

In this paper, we prove that this approach extends to several 2D games
in the Super Mario Bros.\ platform video game series,%
\footnote{After all, ``Braid is a postmodern Super Mario Bros.''\ 
  \cite{braid-forbes}.}
when we generalize them to remove any limits on time,
the total number of enemies, or the number of enemies at each location.
Most of these games (even the original Super Mario Bros.\ from 1985)
have mechanics for spawning arbitrarily many enemies over time,
but it is difficult to find the right combination of mechanics
to implement the specific functionality of an
increment/decrement/jump-if-zero counter.
Specifically, we show that the following games are RE-complete
by building computer machines where the number of enemies
in a particular location represents the value of each counter:
\begin{enumerate}
\item The \emph{New Super Mario Bros.}\ series (Section~\ref{sec:NSMB}):
  New Super Mario Bros., New Super Mario Bros.\ Wii,
  New Super Mario Bros.\ U, and New Super Mario Bros.\ 2.
  One reduction covers all four games,
  using their powerful ``event'' game mechanic,
  where Mario's location can toggle the existence of blocks.
  We build an entire universal counter machine within a single screen
  of the Wii version, so solving even a single-screen level is RE-complete.

\item The \emph{Super Mario Maker} series (Section~\ref{sec:SMM}):
  Super Mario Maker 1 in all four game styles
  (Super Mario Bros.\ 1, Super Mario Bros.\ 3, Super Mario World,
  and New Super Mario Bros.\ U)
  and Super Mario Maker 2 in all five game styles
  (the same four, plus Super Mario 3D World).
\end{enumerate}

All of these games are in RE: if a level is solvable, then there is a finite
algorithm to solve them, by trying all possible sequences of inputs to the
game, and simulating the result.  Because any solvable level is (by definition)
solvable in finite time, and the state in the game after any finite time
is finite, this algorithm is finite.  Thus, to prove RE-completeness,
it remains to prove RE-hardness.

To simplify the process of proving such games RE-hard using counter machines,
we develop in Section~\ref{sec:Counter Gadgets}
a new theory of ``counter gadgets'' which shows that
a \emph{single} gadget diagram suffices to prove RE-hardness
in a one-player game like Super Mario Bros.
This framework builds on the \defn{motion-planning-through-gadgets framework}
developed in recent years
\cite{Gadgets_ITCS2020,Doors_FUN2020,GadgetsIO_WALCOM2022,GadgetsVictory_WALCOM2022,GadgetsChecked_FUN2022,GadgetsVictory_WALCOM2022,lynch2020framework,hendrickson2021gadgets},
starting with FUN 2018 \cite{Toggles_FUN2018}.
In the single-player version of this framework, one agent (Mario)
traverses an environment consisting of local ``gadgets'',
with the locations of the gadgets connected together in a graph
(representing freely traversable connections).
Each \defn{gadget} has a finite set of states and a finite set of
possible transitions, where a transition $(q, a) \to (r, b)$ means that,
when the gadget is in state~$q$, the agent can enter at location $a$
and leave at location~$b$, while changing the state of the gadget to~$r$.
We generalize this framework to allow for \emph{infinite-state} gadgets
called \defn{counter gadgets}, where the states are the natural numbers
and the traversals are among common operations such as increment,
decrement, traversable-if-zero, and traversable-if-nonzero.
Specifically, we show that any one of the following counter gadgets
suffices to prove RE-completeness of the \defn{reachability} problem
(can the agent get from one location to another?):
\begin{enumerate}
\item \textbf{Inc-Dec-JZ:} One traversal path that increments the counter value,
  another that decrements the counter value unless it is already zero,
  and a third that leads the agent to two different locations depending
  on whether the counter value is zero.
\item \textbf{Inc-JZDec:} One traversal path that increments the counter value,
  and another that leads the agent to two different locations depending
  on whether the counter value is zero, and if it is not, decrements the
  counter.
\item \textbf{Inc-DecNZ-PZ:}
  One traversal path that increments the counter value,
  another that decrements the counter value but which can be traversed
  only if the counter is nonzero,
  and a third that is traversable only if the counter value is zero.

  Notably, this RE-hardness result holds even when the connections between
  gadgets form a planar graph, so there is no need for a crossover gadget.
\item \textbf{\boldmath Inc$[a,b]$-DecNZ$[c,d]$-PZ:}
  A generalization of the previous gadget where, when the agent traverses
  an increment or decrement path, it gets to choose how much to increment
  or decrement, within a range of $[a,b]$ or $[c,d]$, respectively,
  where $a,c > 0$.
  This robustness to unknown gadget behavior
  is helpful for building gadgets in video games,
  which can require very careful timing and alignment to force a
  specific number of increments or decrements,
  but forcing at least one and at most some constant is relatively easy.
\end{enumerate}

Specifically, we build an Inc$[1,2]$-DecNZ$[1,2]$-PZ gadget in
the New Super Mario Bros.\ series, and we build three different
Inc-DecNZ-PZ gadgets for different variants of the Super Mario Maker series.
Some of our gadgets are available for download and play \cite{github}.
On the plus side, we need to build only one main gadget for each game,
together with a crossover gadget in some cases (such as New Super Mario Bros.).
On the negative side, as we will see, that one main gadget is quite complex.

By contrast, the Braid proof \cite{braid} had six individual gadgets.
We can instead combine three of them (Lever Pull, Counter, and Branch)
in a straightforward way to build an Inc-JZDec gadget;
see Figure~\ref{fig:braid}.
Then we only need a crossover gadget for the agent (Tim),
which is one of the three crossover gadgets in \cite{braid};
the other two are no longer necessary.

\begin{figure}[t]
  \centering
  \includegraphics[width=\linewidth]{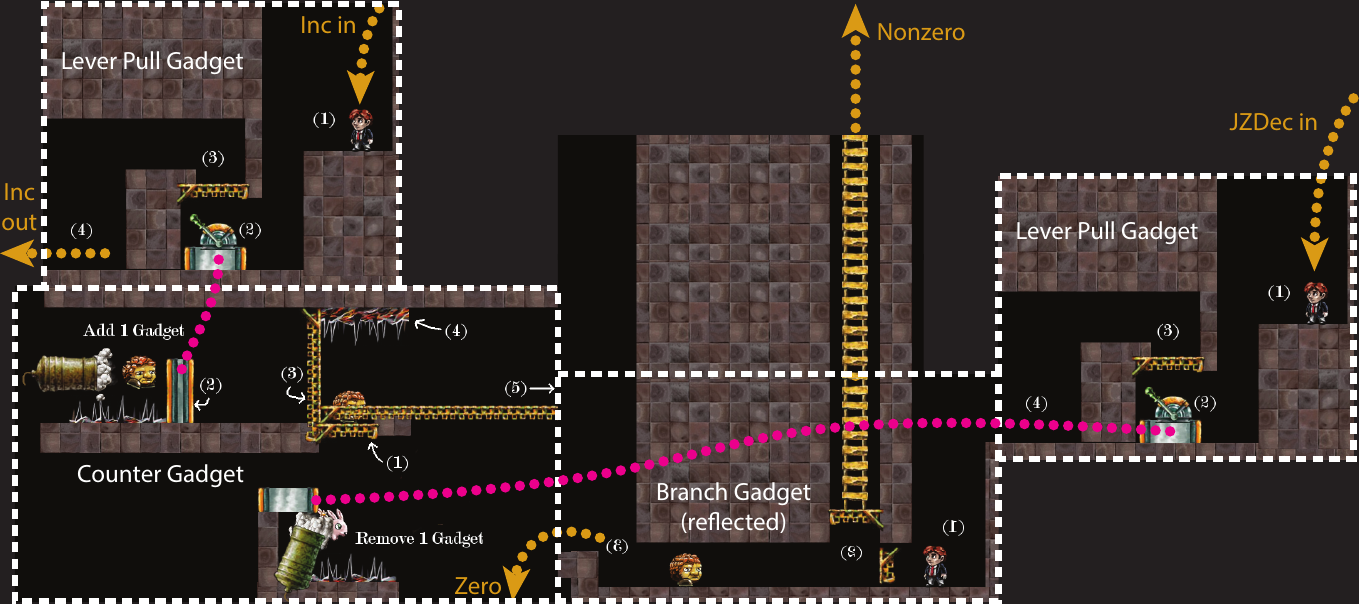}
  \caption{An Inc-JZDec counter gadget in Braid, built from
    the Level Pull, Counter, and Branch gadgets of \cite{braid}.}
  \label{fig:braid}
\end{figure}

Like Hamilton, we conjecture that our counter-gadget framework
can be applied to prove RE-hardness of other video games as well.
We discuss further possibilities for Super Mario Bros.\ games
in Section~\ref{sec:open}.

\section{Counter Gadgets}
\label{sec:Counter Gadgets}

In this paper, we reduce from ``reachability with gadgets'',
first explored in \cite{Toggles_FUN2018,Gadgets_ITCS2020}.
We define a \defn{gadget} $G = (Q, L, T)$ to consist of
a set $Q$ of \defn{states} (not necessarily finite),
a finite set $L$ of \defn{locations},
and a set $T \subseteq (Q\times L)^2$ of allowed \defn{transitions}
on pairs of locations and states,
each written in the form $(q,a) \to (r,b)$
where $q,r \in Q$ and $a,b \in L$.

An example of a gadget is the \defn{1-toggle}: it has a single path that the player can cross in only one direction, and every time they do, the allowed direction flips. \figref{1-toggle} gives a graphical representation. In this case, there are two locations $L = \{a, b\}$, two states $Q = \{R, L\}$, and $T = \{(0, a)\to(1, b), (1, b)\to(0, a)\}$.

\begin{figure}[t]
    \centering
    \includegraphics[scale=1.0]{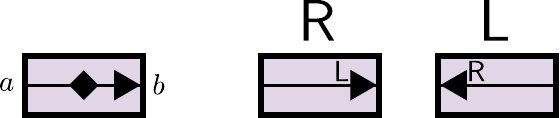}
    \caption{A 1-toggle, along with its state diagram on the right. When the player goes from $a$ to $b$, the arrow flips.}
    \label{1-toggle}
\end{figure}

Here, we consider gadgets with an infinite number of states,
namely, one for each natural number.
We further restrict a \defn{counter gadget}
to consist of \defn{counter components}
(see \figref{gadget-components})
that interact with each other when put in the same gadget:

\begin{figure}[t]
    \centering
    \includegraphics[width=\linewidth]{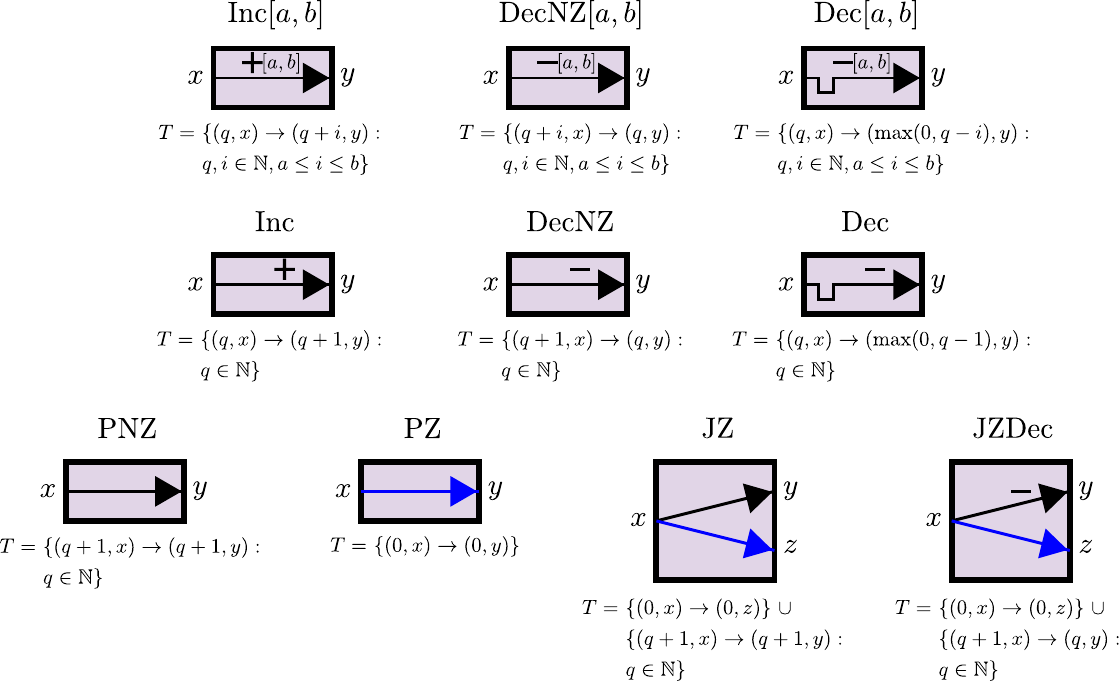}
    \caption{Counter components that are allowed in counter gadgets, along with their sets $T$ of allowed traversals.}
    \label{gadget-components}
\end{figure}

\begin{itemize}
    \item \textbf{\boldmath Inc$[a, b]$}. This is a directed tunnel that is always traversable. When the player traverses it, they choose a natural number $i$ such that $a\le i\le b$. The gadget's state increments by~$i$.
    \item \textbf{\boldmath DecNZ$[a, b]$}. This is a directed tunnel that is traversable if and only if the gadget's state is at least $a$. When the player traverses it, they choose a natural number $i$ such that $a\le i\le \min(s, b)$, where $s$ is the gadget's state. The gadget's state decrements by $i$.
    \item \textbf{\boldmath Dec$[a, b]$}. This is like DecNZ$[a, b]$, except that it is always traversable, and if the gadget's state would become negative, it instead becomes 0.
    \item \textbf{Inc}, \textbf{DecNZ}, \textbf{Dec}. These are short for Inc$[1, 1]$, DecNZ$[1, 1]$, and Dec$[1, 1]$, respectively.
    \item \textbf{PZ}. This is a directed tunnel that is traversable if and only if the gadget's state is 0. It does not change the state.
    \item \textbf{PNZ}. This is a directed tunnel that is traversable if and only if the gadget's state is not 0. It does not change the state. This is only defined for convenience of defining the JZ switch below.
    \item \textbf{JZ}. This is a switch formed by putting a PZ tunnel and a PNZ tunnel in the same gadget and combining their entrances.
    \item \textbf{JZDec}. This is a switch formed by putting a PZ tunnel and a DecNZ tunnel in the same gadget and combining their entrances.
\end{itemize}

\subsection{Undecidable Gadgets via Counter Machines}

A \defn{system} of gadgets consists of instances of gadgets,
an initial state for each gadget,
and a graph connecting the gadgets' locations together.
In the \defn{reachability} problem, we are given a system of gadgets,
a start location $s$, and a goal location $t$,
and we want to know whether the player can get from $s$ to $t$
by making a sequence of gadget traversals and following edges of the
connection graph.

With infinite-state gadgets, sometimes reachability is undecidable, since we can simulate counter machines. We use the fact that the counter-machine halting problem is RE-hard \cite{counter} as a starting point. First, we give a brief introduction to counter machines.

A \defn{counter machine} is a set of counters and a sequence of instructions that run on those counters. The instructions are:
\begin{itemize}
    \item $\vb{Inc(}c\vb{)}$: Add 1 to counter $c$ and move on to the next instruction
    \item $\vb{Dec(}c\vb{)}$: Subtract 1 from counter $c$, leaving its value alone if it was 0, and move on to the next instruction.
    \item $\vb{JZ(}c\vb{,}i\vb{)}$: Check if counter $c$ is 0. If so, jump to instruction $i$. Otherwise, move on to the next instruction.
    \item $\vb{Halt}$: Stop the machine.
\end{itemize}
It is RE-hard to determine whether a counter machine
reaches a $\vb{Halt}$ instruction,
even with just two counters \cite{counter}, \cite[pp.~255--258]{Minsky-book}.

\subsection{Inc-Dec-JZ is RE-complete}

The Inc-Dec-JZ gadget (\figref{inc-dec-jz}) consists of, as the name indicates, an Inc tunnel, a Dec tunnel, and a JZ switch.
We prove that reachability with this gadget is NP-hard.
The proof involves duplicating the Inc and Dec tunnels a bunch, then duplicating the JZ switch a bunch, then simulating a counter machine using multiple of the gadget, using each set of Inc, Dec, and JZ components as an instruction. We will show RE-hardness again, but with a different method where we build a gadget made for flow control, and use multiple copies of that gadget along with unaltered Inc-Dec-JZ gadgets.

\begin{figure}
    \centering
    \includegraphics[scale=1.0]{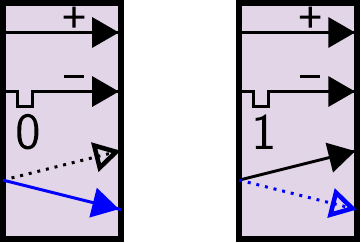}
    \caption{The Inc-Dec-JZ gadget, shown in states 0 and 1.}
    \label{inc-dec-jz}
\end{figure}

\begin{theorem}
    Reachability with Inc-Dec-JZ is RE-complete.
\end{theorem}
\begin{proof}
    We show this by simulating a counter machine with $\vb{Inc(}c\vb{)}$, $\vb{Dec(}c\vb{)}$, and $\vb{JZ(}c\vb{,}i\vb{)}$ instructions, which increment counter $c$, decrement counter $c$ (saturated at 0), and jump to instruction $i$ if $c=0$, respectively.

    First, we build an \defn{Inc-DecNZ-DecNZ} gadget to help with flow control. We simulate one as shown in \figref{inc-decnz-decnz}. If the agent enters via $\textrm{Inc in}$, they increment the top gadget and leave. If the agent enters via $\textrm{DecNZ}_i\textrm{ in}$, they eventually get stuck if the top counter is 0. Otherwise, they increment the middle/bottom gadget, decrement the top gadget, and decrement the gadget that they incremented before, leaving via $\textrm{DecNZ}_i\textrm{ out}$.

\begin{figure}
    \centering
    \includegraphics[scale=0.75]{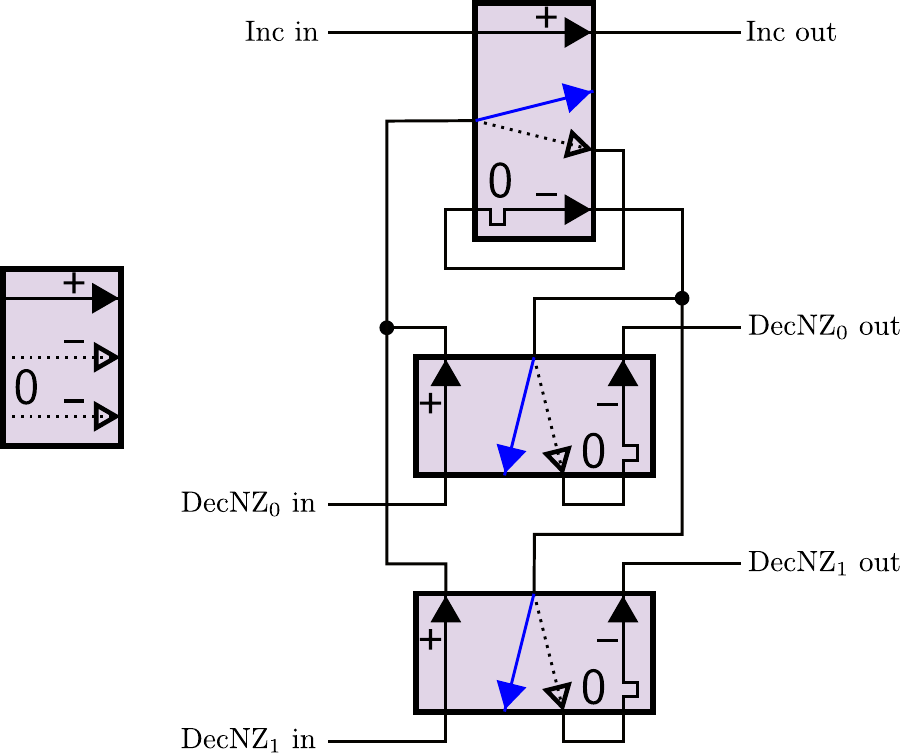}
    \caption{The Inc-DecNZ-DecNZ gadget, along with its simulation by Inc-Dec-JZ gadgets.}
    \label{inc-decnz-decnz}
\end{figure}

    To construct the counter machine, we use an Inc-Dec-JZ gadget $c$ for each counter $c$, and an Inc-DecNZ-DecNZ gadget $i$ for each instruction $i\vb{:}\cdots$. At a high level, we connect gadgets so that the agent flow works as follows:
    \begin{itemize}
        \item For an instruction $i\vb{:Inc(}c\vb{)}$, the agent increments instruction gadget $i$, increments the counter gadget $c$, finds the incremented instruction gadget $i$ and decrements it, and then moves on to the next instruction gadget $i+1$.
        \item For an instruction $i\vb{:Dec(}c\vb{)}$, the agent does the same as above, except that they decrement the counter gadget $c$.
        \item For an instruction $i\vb{:JZ(}c\vb{,}i'\vb{)}$, the agent increments instruction gadget $i$ and goes to check whether counter gadget $c$ is in state 0. If it is, they decrement gadget $i$ and branch to instruction gadget $i'$. Otherwise, they decrement gadget $i$ using the other decrement path and move on to the next instruction gadget $i+1$.
    \end{itemize}
    \figref{inc-dec-jz-counter-machine} shows an example.

\begin{figure}
    \centering
    \includegraphics[width=\linewidth]{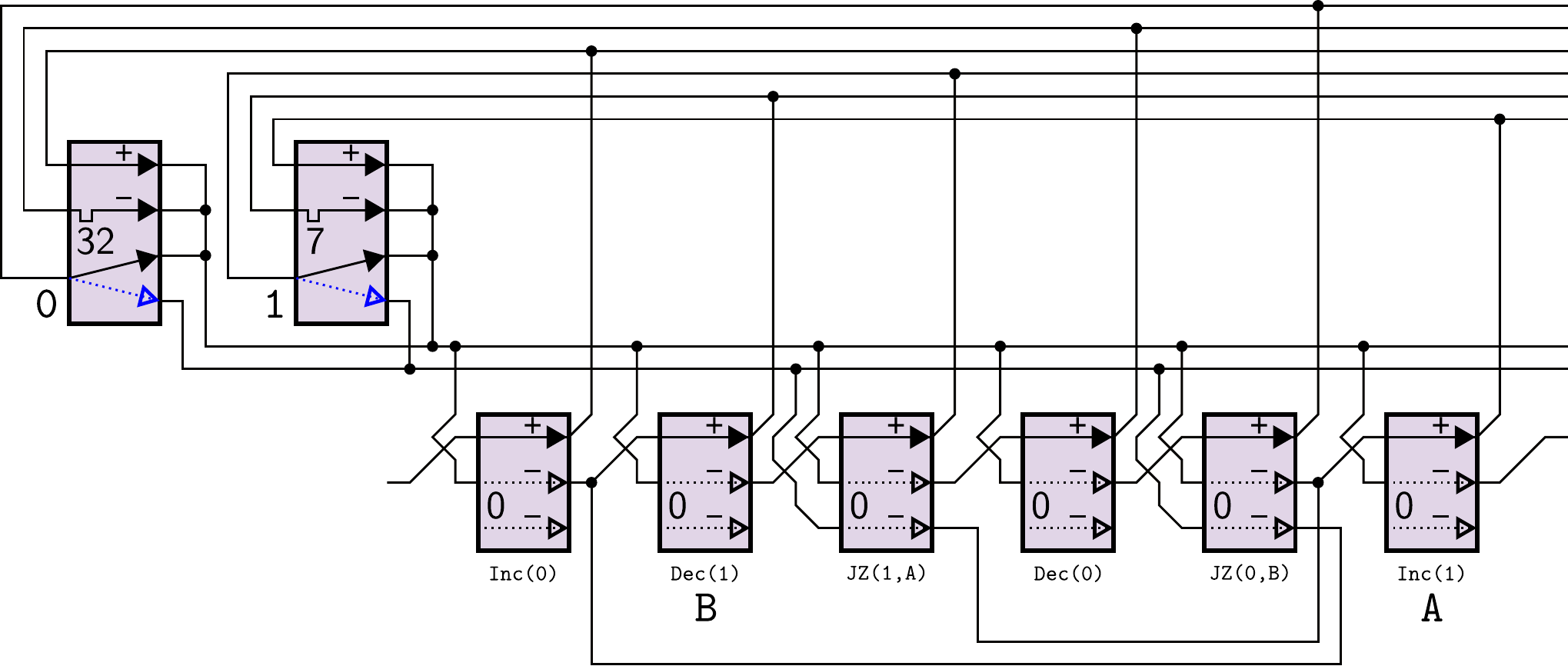}
    \caption{The Inc-Dec-JZ gadget simulating a counter machine.}
    \label{inc-dec-jz-counter-machine}
\end{figure}

    More concretely, we connect gadgets in the following ways, where $a[b]$ denotes location $b$ of gadget $a$, $I$ denotes the counter machine program, $I[i]$ denotes the $i$th instruction, and $C$ denotes the set of counter gadgets:
    \begin{itemize}
        \item $i[\textrm{Inc out}]\sim c[\textrm{Inc in}]$ for all $i, c$ where $I[i] = \vb{Inc(}c\vb{)}$.
        \item $i[\textrm{Inc out}]\sim c[\textrm{Dec in}]$ for all $i, c$ where $I[i] = \vb{Dec(}c\vb{)}$.
        \item $i[\textrm{Inc out}]\sim c[\textrm{JZ in}]$ for all $i, c, i'$ where $I[i] = \vb{JZ(}c\vb{,}i'\vb{)}$.
        \item $c[\textrm{Inc out}]\sim i[\textrm{DecNZ}_0\textrm{ in}]$ for all $i, c$ where $0\le i < |I|$ and $c\in C$.
        \item $c[\textrm{Dec out}]\sim i[\textrm{DecNZ}_0\textrm{ in}]$ for all $i, c$ where $0\le i < |I|$ and $c\in C$.
        \item $c[\textrm{JZ out }(\ne 0)]\sim i[\textrm{DecNZ}_0\textrm{ in}]$ for all $i, c$ where $0\le i < |I|$ and $c\in C$.
        \item $c[\textrm{JZ out }(= 0)]\sim i[\textrm{DecNZ}_1\textrm{ in}]$ for all $i, c$ where $0\le i < |I|$ and $c\in C$.
        \item $i[\textrm{DecNZ}_0\textrm{ out}]\sim (i+1)[\textrm{Inc in}]$ for all $i$ where $0\le i < |I|-1$.
        \item $i[\textrm{DecNZ}_1\textrm{ out}]\sim i'[\textrm{Inc in}]$ for all $i, c, i'$ where $I[i] = \vb{JZ(}c\vb{,}i'\vb{)}$.
    \end{itemize}
    
    The agent starts just in front of the first instruction gadget, and any instruction gadget that corresponds to a $\vb{Halt}$ instruction is replaced with a goal. Then the agent can win if and only if the counter machine halts, reducing the counter machine halting problem to reachability.
\end{proof}

\subsection{Inc-JZDec is RE-complete}

The Inc-JZDec gadget (\figref{inc-jzdec}) replaces the Dec and JZ components with a single JZDec switch.
We prove that reachability with this gadget is RE-hard by simulating the Inc-Dec-JZ gadget.

\begin{figure}
    \centering
    \includegraphics[scale=1.0]{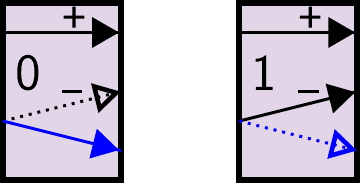}
    \caption{The Inc-JZDec gadget, shown in states 0 and 1.}
    \label{inc-jzdec}
\end{figure}

\begin{theorem}
    Reachability with Inc-JZDec is RE-complete.
\end{theorem}
\begin{proof}
    We reduce from reachability with Inc-Dec-JZ by simulating the Inc-Dec-JZ gadget, as shown in \figref{inc-jzdec-sim-inc-dec-jz}. We use $G_0$ and $G_1$ to store the counter's value. If the agent enters via $\textrm{Inc in}$, they increment $G_0$ and $G_1$ and leave via $\textrm{Inc out}$. If the agent enters via $\textrm{Dec in}$, then if the counter is 0, they nearly immediately leave via $\textrm{Dec out}$. Otherwise, they decrement $G_0$, increment $H_2$, decrement $G_1$, decrement $H_2$, and leave via $\textrm{Dec out}$. Note that $H_2$ has no net change. If the agent enters via $\textrm{JZ in}$, if the counter is 0, they nearly immediately leave via $\textrm{JZ out }=0$. Otherwise, they decrement $G_1$, increment $H_1$, increment $G_1$, decrement $H_1$, and leave via $\textrm{JZ out }\ne 0$, leaving no net change (except in $H_0$, but that gadget is just used as a diode). So this simulation works.
\end{proof}

\begin{figure}
    \centering
    \includegraphics[scale=0.75]{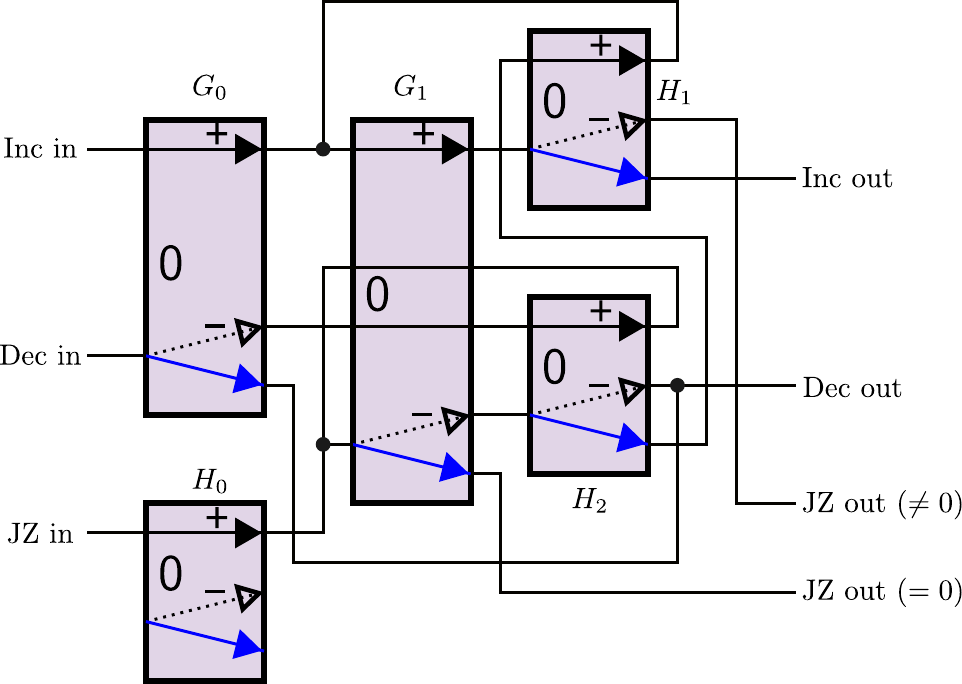}
    \caption{Simulation of Inc-Dec-JZ with Inc-JZDec. Gadgets $G_0$ and $G_1$ store the counter's value, while $H_0$, $H_1$ and $H_2$ are used for flow control.}
    \label{inc-jzdec-sim-inc-dec-jz}
\end{figure}

\subsection{Inc-DecNZ-PZ is RE-complete}

The Inc-DecNZ-PZ gadget (\figref{inc-decnz-pz}) replaces the JZDec switch with separate DecNZ and PZ tunnels.
This gadget can easily simulate the Inc-JZDec gadget,
by combining the entrances of DecNZ and PZ.

\begin{corollary}
    Reachability with Inc-DecNZ-PZ is RE-complete.
\end{corollary}

\begin{figure}
    \centering
    \includegraphics[scale=1.0]{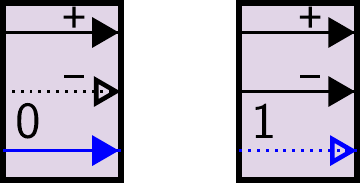}
    \caption{The Inc-DecNZ-PZ gadget, shown in states 0 and 1.}
    \label{inc-decnz-pz}
\end{figure}

In addition, we have a planar result.
In \defn{planar reachability}, we restrict to \defn{planar} systems of gadgets
where the graph of connections between gadgets' locations do not cross
gadgets or each other (except at common endpoints).

\begin{theorem}\label{planar}
    Planar reachability with any planar system of Inc-DecNZ-PZ gadgets is RE-complete.
\end{theorem}
\begin{proof}
    Ani et al.~\cite[Theorems~3.1 and~3.2]{Doors_FUN2020}
    show that a crossover can be built from a
    \defn{symmetric self-closing door}:
    a gadget with two states $Q = \{1, 2\}$
    and two possible traversal paths $L_1 \to R_1$ and $L_2 \to R_2$,
    where $L_1 \to R_1$ is possible only in state $1$,
    $L_2 \to R_2$ is possible only in state $2$,
    and every traversal switches the state.
    \figref{inc-decnz-sim-sscd-planar}
    shows how to build a symmetric self-closing door from Inc-DecNZ
    (and thus Inc-DecNZ-PZ) in all cases.
\end{proof}

\begin{figure}
    \centering
    \includegraphics[scale=1.0]{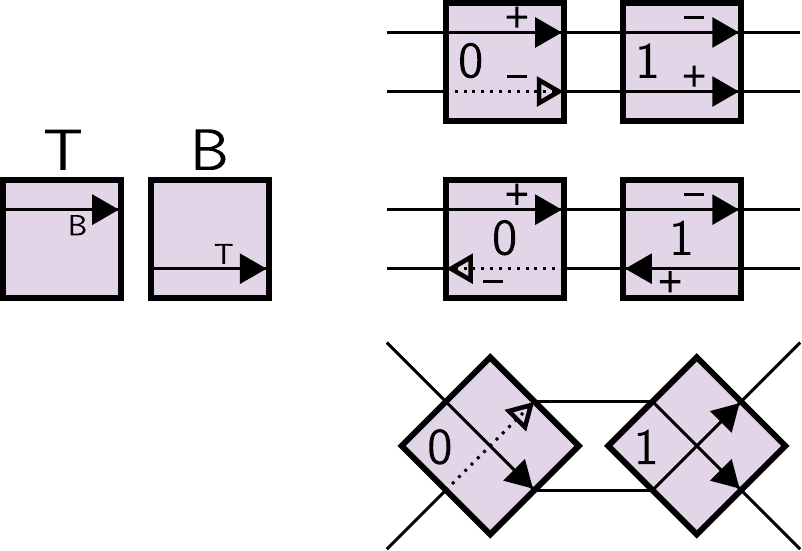}
    \caption{Left: State diagram of the symmetric self-closing door. Right: Simulation of a symmetric self-closing door, no matter the planar arrangement of tunnels in Inc-DecNZ.}
    \label{inc-decnz-sim-sscd-planar}
\end{figure}

Thus, we do not need to build a crossover when reducing from reachability with the Inc-DecNZ-PZ gadget, even in a planar application.

\subsection{Inc\texorpdfstring{$[a,b]$}{[a,b]}-DecNZ\texorpdfstring{$[c,d]$}{[c,d]}-PZ, \texorpdfstring{$a > 0$, $c > 0$}{a > 0, c > 0} is RE-complete}

The Inc$[a,b]$-DecNZ$[c,d]$-PZ gadget (\figref{incab-decnzcd-pz}) replaces the Inc and DecNZ tunnels with Inc$[a,b]$ and DecNZ$[c,d]$ tunnels, respectively. Recall that an Inc$[a,b]$ tunnel allows the player to choose an integer between $a$ and $b$ inclusive, and increment the gadget's state by that amount. A DecNZ$[c,d]$ tunnel allows the player to choose an integer between $c$ and $d$ inclusive and decrement the state by that integer, but only if the result is nonnegative.

\begin{figure}
    \centering
    \includegraphics[scale=1.0]{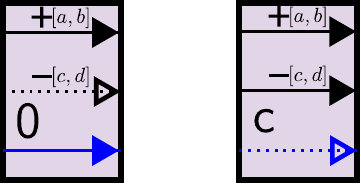}
    \caption{The Inc$[a,b]$-DecNZ$[c,d]$-PZ gadget, shown in states 0 and $c$.}
    \label{incab-decnzcd-pz}
\end{figure}

\begin{theorem}\label{thm:incab-decnzcd-pz}
    Reachability with Inc$[a,b]$-DecNZ$[c,d]$-PZ is RE-hard if $a > 0$ and $c > 0$.
\end{theorem}
\begin{proof}
We simulate the Inc-DecNZ-PZ gadget. First, we build an \defn{edge duplicator} (\figref{incab-decnzcd-pz-edge-dup}), which allows us to duplicate Inc$[a,b]$ and DecNZ$[c,d]$ tunnels as many times as we want. In the edge duplicator shown, $e_0\to e_1$ is being duplicated. If the player enters via ${\rm In}_0$, they go from $f$ to $g$ $c$ times and go through $e_0$ to reach $e_1$. Then since the gadget on the right does not allow passage, they must go from $h$ to $i$ until the path to ${\rm Out}_0$ opens up (that is, $a$ times). Then they leave, and the left gadget is reset. This works symmetrically for ${\rm In}_1\to{\rm Out}_1$.

\begin{figure}
    \centering
    \includegraphics[scale=0.75]{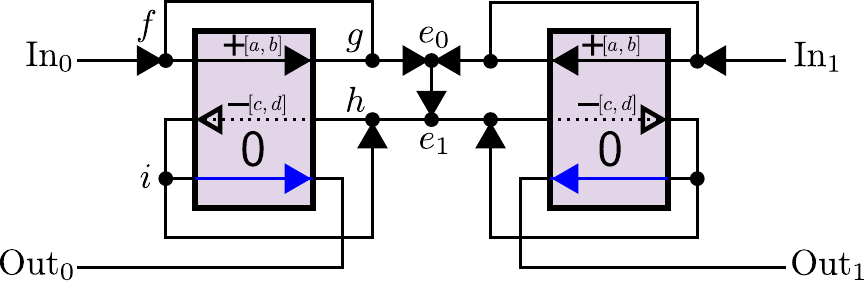}
    \caption{An edge duplicator built with the Inc$[a,b]$-DecNZ$[c,d]$-PZ gadget. The player can go from $e_0$ to $e_1$ along two different paths without leaking between them.}
    \label{incab-decnzcd-pz-edge-dup}
\end{figure}

Then we use two gadgets with as many Inc$[a,b]$ tunnels, as many DecNZ$[c,d]$ tunnels, and as many PZ lines as we want to simulate the Inc-DecNZ-PZ gadget (\figref{incab-decnzcd-pz-sim-inc-decnz-pz}). In fact, we use $acd$ Inc$[a,b]$ tunnels and $abd$ DecNZ$[c,d]$ tunnels in $G_0$, and $bcd$ Inc$[a,b]$ tunnels and $abc$ DecNZ$[c,d]$ tunnels in $G_1$.

Let $s(G)$ be the state of gadget $G$. We say that $G\sqsubset[i, j]$ if $s(G)$ is between $i$ and $j$, and it is possible that $s(G) = i$, and also possible that $s(G) = j$. For example, if $G$ starts at state 0, then $G\sqsubset[0, 0]$. If the player goes through the Inc$[a,b]$ tunnel, then $G\sqsubset[a, b]$. After going through the DecNZ$[c,d]$ tunnel, which is only possible if $b - c\ge 0$, $G\sqsubset[\max(a - d, 0), b - c]$. If $G\sqsubset[i, j]$, we define $\min(G) = i$ and $\max(G) = j$. 

We maintain the invariant that $\max(G_0) = \min(G_1) = abcdn$, where $n$ is state of the simulated Inc-DecNZ-PZ gadget.
\begin{itemize}
    \item If the player crosses the simulated Inc tunnel, then $\max(G_0) := \max(G_0) + b\cdot acd = abcd(n + 1)$ and $\min(G_1) := \min(G_1) + a\cdot bcd = abcd(n + 1)$.
    \item If the player successfully crosses the simulated DecNZ tunnel, then because of $G_0$, it was the case that $abcdn \ge c\cdot abd$, meaning that $n > 0$, which is what we want. Then $\max(G_0) := \max(G_0) - c\cdot abd = abcd(n - 1)$ and $\min(G_1) := \min(G_1) - d\cdot abc = abcd(n - 1)$. Note that since $\max(G_0) = \min(G_1)$, if $G_0$'s portion was crossable, then so is $G_1$'s portion.
    \item If the player successfully crosses the simulated PZ tunnel, then because of $G_1$, it was the case that $abcdn = 0$, meaning that $n = 0$. Then $\max(G_0) := 0$ and $\min(G_1) := 0$. If $G_1$'s portion was crossable, then so is $G_0$'s portion, because $\max(G_0) = \min(G_1)$.
\end{itemize}

So this simulation works.
\end{proof}

\begin{figure}
    \centering
    \includegraphics[scale=1.0]{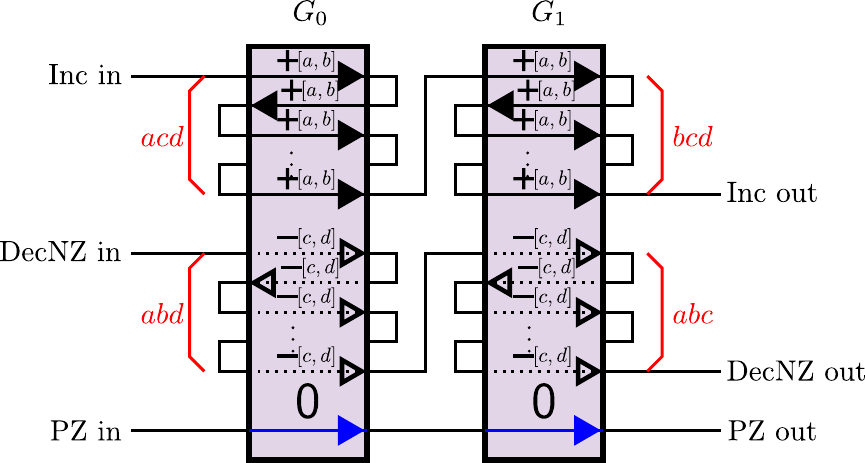}
    \caption{Simulation of Inc-DecNZ-PZ with Inc$[a,b]$-DecNZ$[c,d]$-PZ gadget. The numbers of copies of repeated tunnels are shown in red.}
    \label{incab-decnzcd-pz-sim-inc-decnz-pz}
\end{figure}

\subsection{Constant-Size Levels}
\label{sec:constant}

Universal Turing or counter machines let us strengthen our results to need
just a constant number of gadgets.
For example, Korec \cite{smallCounters} designs a 2-counter machine $U_{32}$
(consisting of 32 instructions over \texttt{Inc}, \texttt{Dec}, and \texttt{JZ})
that is \defn{strongly universal} in the sense that
there is a computable function $f$ such that, given any 2-counter machine $M$,
$M$ applied to $x$ and $y$ produces the same result as
$U_{32}$ applies to counter values $f(x)$ and $y$.
By applying the theorems above to $U_{32}$,
we obtain a system of a constant number of counter gadgets
that simulates $U_{32}$ and thus any 2-counter machine
and thus any Turing machine.
Crucially, this system must start with arbitrary specified initial states,
to represent $f(x)$ and $y$. This framework and its implications are described in more depth in \cite{baba}.

Applied to Mario, this means that we can prove RE-hardness of
constant-size levels, provided they can start with arbitrarily many
enemies at each location.

Alternatively, if we must start with all counters in state 0,
or Mario levels without any enemies, then we need a linear number of
instructions to build up the initial counter values
(by repeated addition and multiplication implemented by repeated addition).

\section{New Super Mario Bros.\ Series}
\label{sec:NSMB}

An earlier version of this work \cite{uuu}
showed that Generalized New Super Mario Bros.\ is undecidable by a reduction from reachability with the Inc-DecNZ-PZ gadget and a crossover.
In this section, we will prove undecidability for all games in the New Super Mario Bros.\ series by building a similar counter with the Inc$[1,2]$-DecNZ$[1,2]$-PZ gadget. Specifically, this gadget can be used in New Super Mario Bros., New Super Mario Bros.\ Wii, New Super Mario Bros.\ U, and New Super Mario Bros.\ 2.
We use the number of enemies called Goombas to keep track of the counter state, and since we can simulate the Inc-DecNZ-PZ gadget by Theorem \ref{incab-decnzcd-pz-sim-inc-decnz-pz} we can obtain undecidability in constant size as in Section \ref{sec:constant}. However, this result is not necessarily planar, so we provide a simple crossover gadget as well.

In this section, we make heavy use of a mechanic in the New Super Mario Bros.\ games called \defn{events}, which allows the player to interact with blocks via switches and event controllers.
The functionality of events that we will use can be summarized as follows, although exact implementations vary per game: 
\begin{itemize}
    \item Each event can be either on or off.
    \item Event controllers toggle one event based on the status of another.
    \item Location controllers toggle an event based on the status of enemies or players in a pre-defined location.
    \item There are blocks which appear or disappear according to the status of an event.
\end{itemize}

\begin{theorem}\label{NSMBUndecidable}
    The New Super Mario Bros.\ games are RE-complete.
\end{theorem}
\begin{proof}

Our counter gadget is composed of two parts: the counter, pictured in Figure~\ref{nsmbc}, and the tunnel, pictured in Figure~\ref{nsmbt}. Each purple area represents a defined location. The groups of blocks associated with a letter are controlled by the event with the given letter\footnote{In memory, events are numbered, not lettered, but we use letters to disambiguate with locations}. When the event with the corresponding letter changes state, the blocks also change state between visible and invisible (where invisible blocks are also intangible). Invisible block are depicted as outlined instead of filled. There are also several invisible control objects:
\begin{itemize}
\item An enemy spawner spawns Goombas \ICON{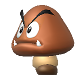}\ from the pipe \ICON{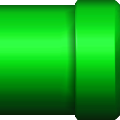}.
    \item A location controller activates event \textsf{A} when a player enters location \textsf{1}.
    \item An event controller activates event \textsf{B} when event \textsf{A} is activated.
    \item A location controller deactivates event \textsf{A} when a player enters location \textsf{2}.
    \item A location controller activates event \textsf{D} when a player enters location \textsf{3}.
    \item An event controller activates event \textsf{E} when event \textsf{D} is activated.
    \item A location controller deactivates event \textsf{D} when a player enters location \textsf{4}.
    \item A location controller activates event \textsf{C} when an enemy enters location \textsf{5}.
    \item An event controller deactivates event \textsf{B} when event \textsf{C} is activated.
    \item A location controller activates event \textsf{F} when an enemy enters location \textsf{7}.
    \item An event controller deactivates event \textsf{E} when event \textsf{F} is activated.
    \item An event controller activates event \textsf{7} if and only if an enemy is in location \textsf{6}.
\end{itemize}

\textbf{Inc$[1,2]$.}\quad When the player enters the Inc tunnel, they encounter location \textsf{1} which enables event \textsf{A}, preventing backtracking and opening the path through the tunnel. They then enter location \textsf{2}, which reverses this change, preventing backtracking. At the same time, event \textsf{A} triggers event \textsf{B} which allows a Goomba \ICON{nsmb/goomba.png}\ to pass through the blocks tied to event \textsf{B}. That same Goomba \ICON{nsmb/goomba.png}\ enters location \textsf{5} which indirectly deactivates event \textsf{B}. Because of the way Goombas \ICON{nsmb/goomba.png}\ spread out when moving, only one or two Goombas \ICON{nsmb/goomba.png}\ will pass through during this time. The incremented Goomba(s) \ICON{nsmb/goomba.png}\ end up in location \textsf{6}.

\textbf{DecNZ$[1,2]$.}\quad The DecNZ tunnel works analogously to the Inc tunnel, instead allowing one Goomba \ICON{nsmb/goomba.png}\ to pass through the blocks tied to event \textsf{E}. In addition, if a Goomba \ICON{nsmb/goomba.png}\ is not in location \textsf{6}, i.e. the counter value is zero, event \textsf{G} will be disabled, and blocks tied to \textsf{G} will block traversal through the decrement tunnel, enforcing the NZ condition.

\textbf{PZ.}\quad Similarly to how the NZ condition is enforced, if a Goomba \ICON{nsmb/goomba.png}\ is in location \textsf{6}, event \textsf{G} is active and blocks tied to \textsf{G} block the path through PZ.

\begin{figure}
\centering
\begin{subfigure}[t]{\linewidth}
    \centering
    \includegraphics[scale=.4]{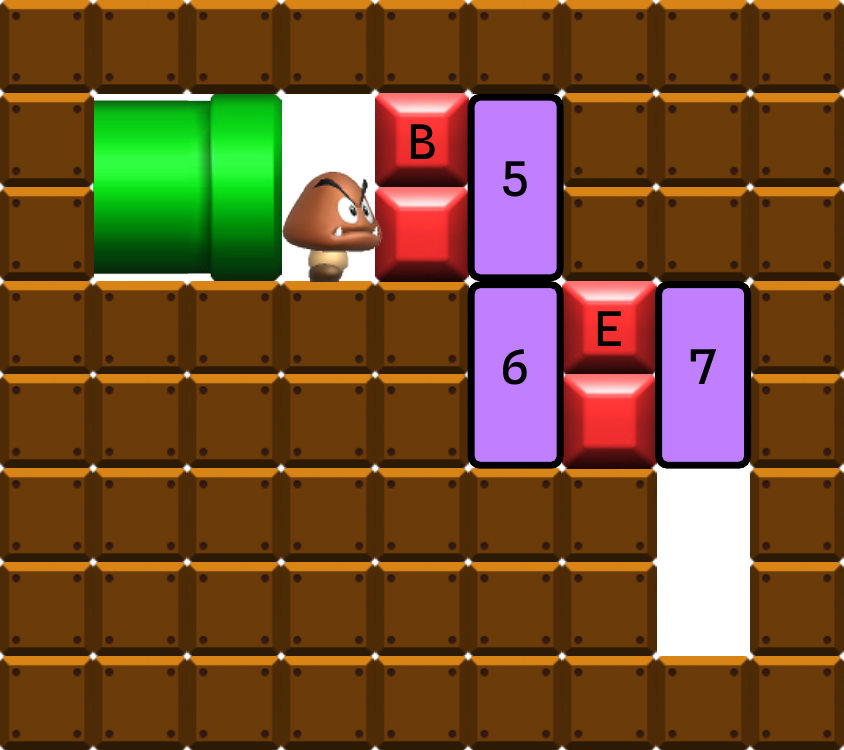}
    \caption{Counter (for Goombas only).}
    \label{nsmbc}
\end{subfigure}

\medskip

\begin{subfigure}[t]{\linewidth}
    \centering
    \includegraphics[scale=.4]{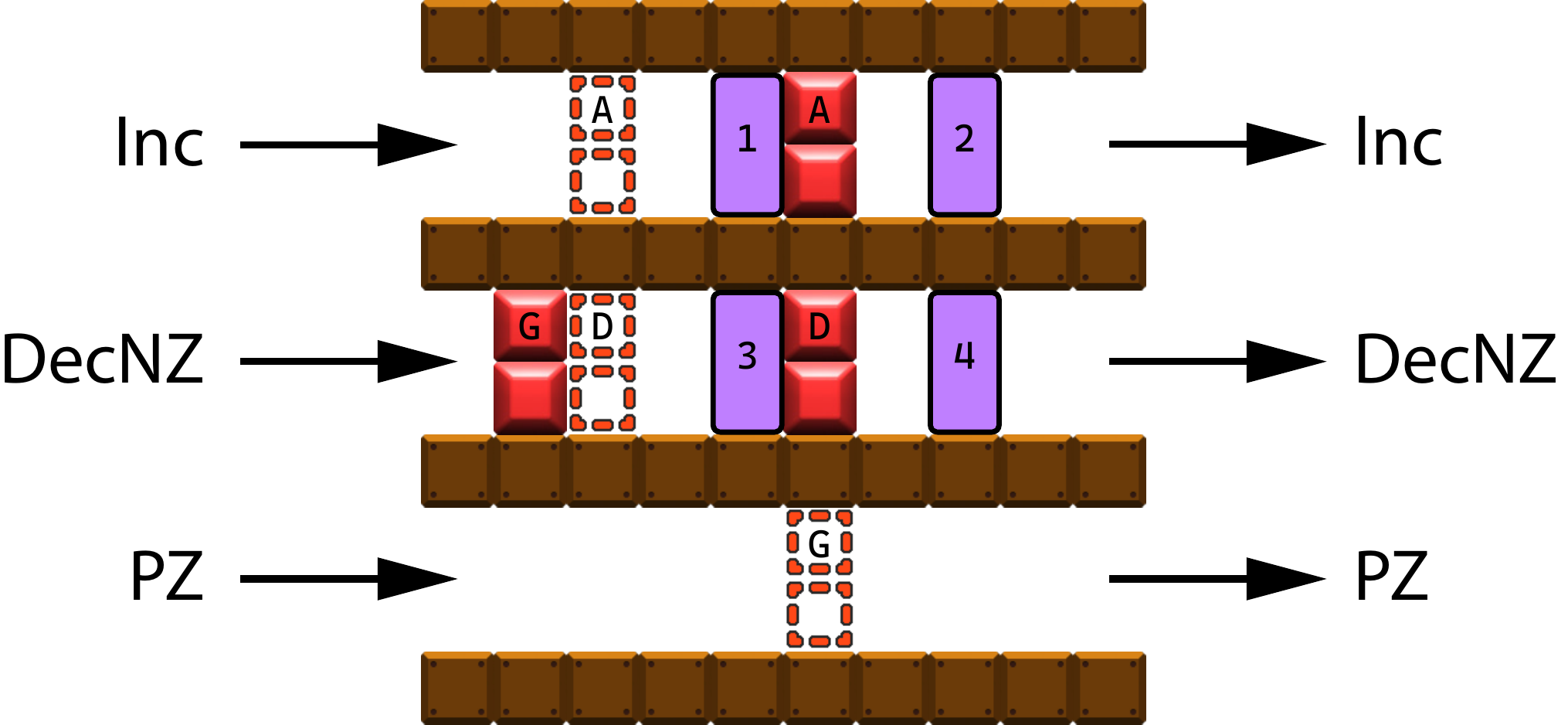}
    \caption{Tunnel (for Mario only).}
    \label{nsmbt}
\end{subfigure}
\caption{The New Super Mario Bros.\ counter gadget in state 0.}
\end{figure}

To complete the proof, we provide a crossover gadget, as pictured in Figure \ref{nsmbCrossover}. This is another event-based gadget, where all blocks are controlled by one event, which is off as pictured. Entering either location labeled \textsf{1} will enable the event, toggling the states of all blocks and allowing traversal only between the top left and bottom right, and entering either location labeled \textsf{2} will disable the event, returning it to the state pictured and allowing traversal between the top right and bottom left.
\end{proof}
  
\begin{figure}
    \centering
    \includegraphics[scale=0.4]{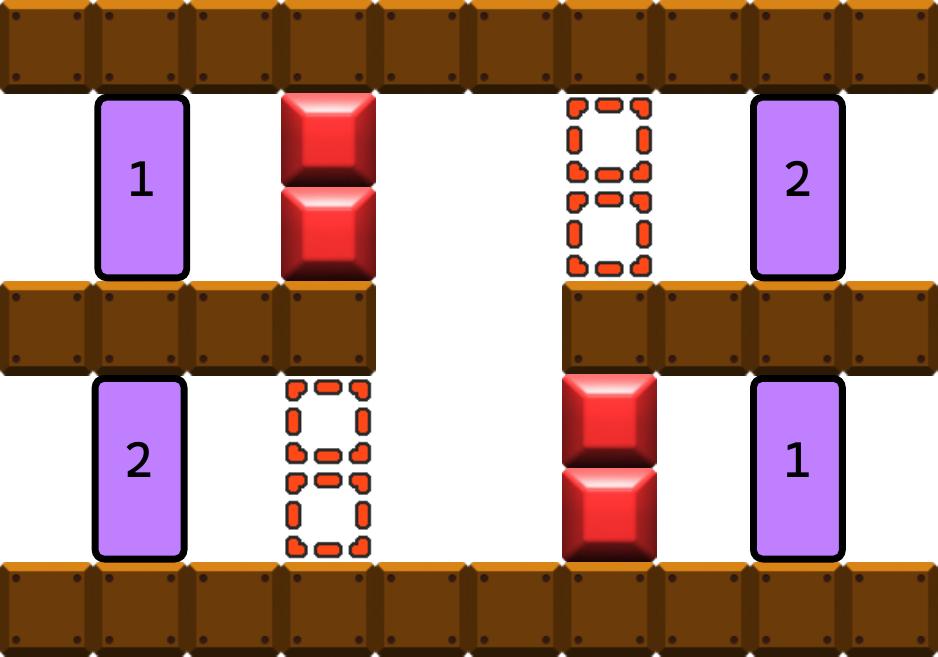}
    \caption{The New Super Mario Bros.\ crossover gadget.}
    \label{nsmbCrossover}
\end{figure}

\subsection{Constant-Size New Super Mario Bros.\ Wii}

As described in Section~\ref{sec:constant},
the above result implies undecidability for a level of constant size.
In fact, we have explicitly built a universal counter machine in
New Super Mario Bros.\ Wii within a \emph{single screen}.
Figure \ref{constantCounterNSMBW} depicts such a counter,
as shown in the \textit{Reggie} level editor.
Our level file is available to download and play \cite{github}.

\begin{figure}
    \centering
    \includegraphics[width=\linewidth]{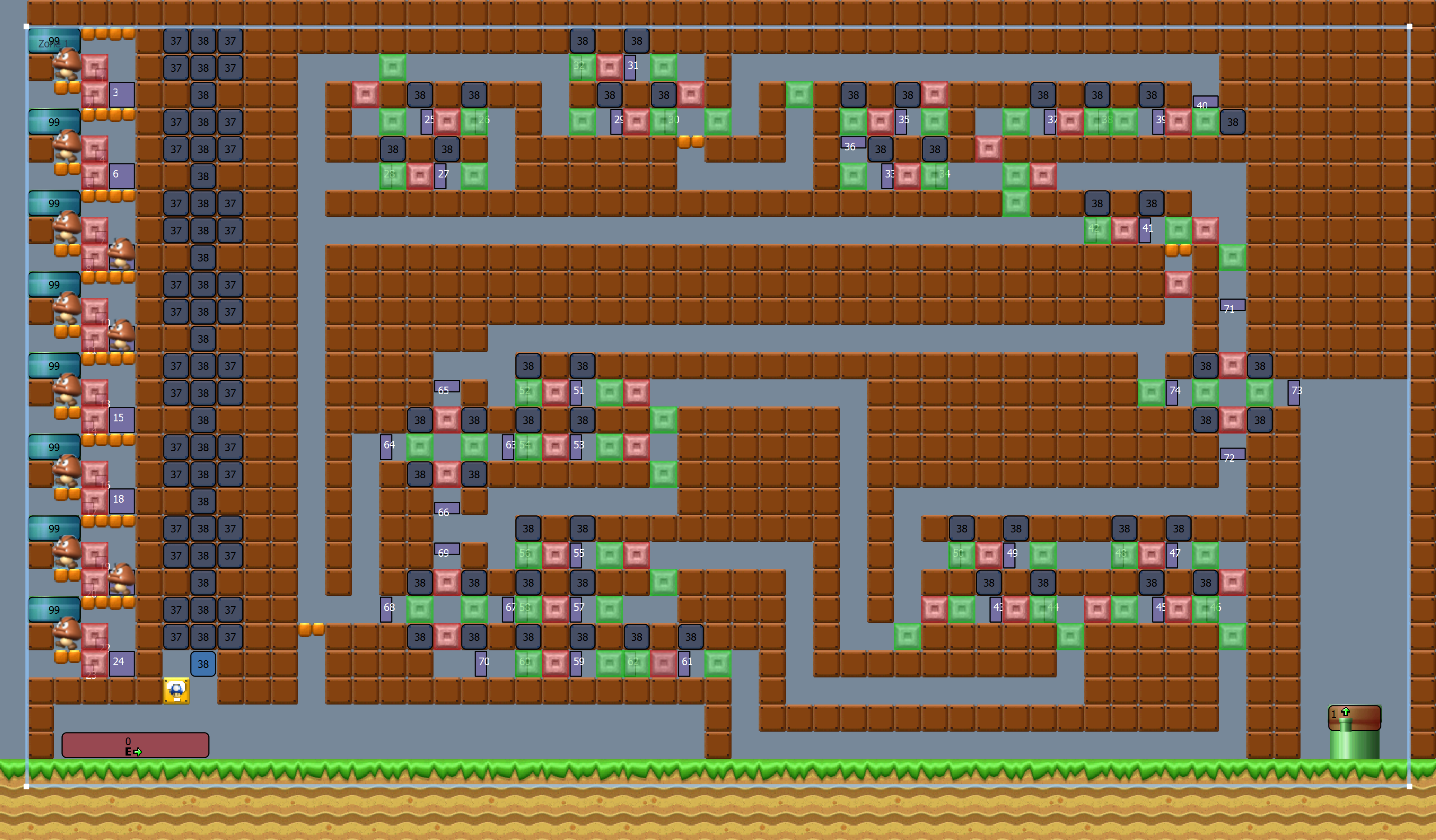}
    \caption{$U_{32}$ in New Super Mario Bros.\ Wii.}
    \label{constantCounterNSMBW}
\end{figure}

Specifically, this level builds the strongly universal counter machine $U_{32}$
of Korec \cite{smallCounters}, following an approach taken for
the video game \emph{Baba Is You} \cite{baba}.
The construction of this level is somewhat different from the reductions
we have described, making many simplifications to fit the level on one screen.
The left side of the level is devoted to the 8 registers, while the majority of the level is devoted to traversal tunnels which control the states of the gadget. Effectively, this is the same as the two components of our gadget featured in Theorem \ref{NSMBUndecidable}, but with extra space removed. Branches are achieved easily by having an event controller to check whether a Goomba is in the counter location. As pictured, the gadget has values of $0,0,1,1,0,0,1,0$ for registers $0,1,2,3,4,5,6,7$ respectively. The goal of the level is to traverse from the starting area (in the bottom left) to the pipe (bottom right) which leads to the flagpole. A mini-mushroom is provided in the starting area to help the player traverse more easily through tight tunnels, but is not essential to the functioning of the gadget.

Because this entire counter fits on one screen, it is unnecessary
to generalize level size as we have with other games.
However, we still need to generalize the game in a few ways.
Specifically, the actual game features a timer,
and the actual game engine will spawn a Goomba from each pipe only if
there are not already eight Goombas on screen.
Both of these limitations need to be removed to make the level fully work.

Because the functionality of events is effectively the same across the
New Super Mario Bros.\ series, this same counter should be able to be built
in the other games, although smaller Nintendo DS screen sizes
seem to prevent making such a large screen size in the DS games.

\section{Super Mario Maker}
\label{sec:SMM}

First we describe how the screen size
(which we do not generalize from the implemented size)
affects local vs.\ global behaviors in Super Mario Maker 1 and 2.
As in most Mario games, memory and effects are typically limited to
the extent of the screen --- technically, the \defn{relevant screen} which
extends four blocks beyond the visible screen.
For example, activating a P-switch temporarily turns coins into blocks,
but only within the current relevant screen; as Mario and this screen moves,
which coins are transformed changes, which can impact nearby enemies etc.
Similarly, enemies typically \defn{spawn} when the relevant screen first
overlaps their start location, and then \defn{despawn} when they leave that
relevant screen.

But the Super Mario Maker game engines also define a notion of
\defn{global} objects \cite{global-loading}
whose state is remembered even when it outside the relevant screen.
Only a handful of objects are inherently global; for example,
One-ways spawn at the beginning of the level and never despawn,
while Yoshi spawns when reaching the screen but then never despawns.
Crucially, an object becomes global if it is on Tracks,
or is on top of another global object;
this principle is called \defn{global ground}.
We use this property to make global any objects we want to be remembered,
such as the enemies representing the state of a counter.

The Super Mario Maker games are also unusual in that they
allow level creation in multiple (4--5) game styles,
which each offer some slightly unique mechanics.
Our constructions that apply to four styles make use of only mechanics
that are present in all four styles,
and all four of the styles have the same physics.

\subsection{Super Mario Maker 1}
\label{sec:smm1}

\begin{theorem}
    All four game styles (Super Mario Bros.\ 1, Super Mario Bros.\ 3, Super Mario World, and New Super Mario Bros.\ U) of Super Mario Maker 1 are RE-complete.
\end{theorem}

\begin{proof}
    We reduce from planar reachability with the Inc-DecNZ-PZ gadget
    (Theorem~\ref{planar}).  Our gadget uses the following elements:

    \begin{itemize}
        \item \textbf{Solid ground:} \ICON{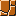}\ \ICON{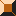}\ Standard blocks with no special effects.
        \item \textbf{Semisolid platforms:} \ICON{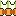}\ \ICON{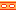}\ \ICON{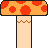}\ Platforms which can be jumped through from below but are solid from above.
        \item \textbf{One-ways:} \ICON{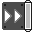}\ Walls which allow entities (Mario, enemies, etc.)\ to pass the white bar only in the direction of the arrows.
        \item \textbf{Brick block:} \ICON{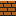}\ Solid block which can be hit from below to defeat enemies or bounce trampolines above it,
        \item \textbf{Coin:} \ICON{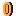}\ Transparent course element which allows Mario and enemies to pass through freely. If Mario touches a coin, it is collected and disappears from the course.
        \item \textbf{P-switches:} \ICON{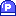}\ Switches which turn coins into brick blocks and vice versa for the duration of a timer.
        \item \textbf{Tracks:} \ICON{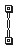}\ Rails specifying periodic movement of attached entities; platforms \ICON{smm/platform}\ on tracks are global ground, meaning that they and entities on them do not despawn when offscreen.
        \item \textbf{Spikes:} \ICON{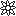}\ Blocks which damage Mario upon contact.
        \item \textbf{Goombas:} \ICON{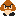}\ Enemies which damage Mario when he contacts from any direction but above. Can safely walk on spikes and are defeated by being jumped on, giving Mario a vertical boost similar to a jump.
        We use big ($2 \times 2$) Goombas in this construction.
        \item \textbf{Trampoline:} \ICON{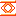}\ Item which bounces entities on it up into the air.
        \item \textbf{Pipes:} \ICON{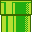}\ Elements which periodically spawn a particular item or enemy (drawn next to the pipe) into the course. If the pipe spawns items, it will only do so if the last element that the pipe spawned is no longer loaded or exists. Enemies such as Goombas \ICON{smm/goomba}\ spawn indefinitely. Their spawn frequency can be adjusted to one of four speeds.
    \end{itemize}

    The key idea in this gadget is similar to other Mario Inc-DecNZ-PZ reductions: use the number of Goombas \ICON{smm/goomba}\ in a particular location to represent the value of a counter. The gadget has infinitely many enumerated states such that, for any integer $n\geq 0$, there exists a state with a collection of $n$ Goombas. The particular construction of the gadget determines how to increase, decrease, and check for zero in the gadget.

    The counter element of the gadget is a semisolid platform \ICON{smm/platform}\ on a track \ICON{smm/track}. The track makes the platform global ground, preventing the Goombas \ICON{smm/goomba}\ from despawning as Mario moves away from the gadget.

    Another critical element to the gadget construction is Goomba pushing. If two Goombas \ICON{smm/goomba}\ are on the same $y$ level and walk into each other, they both bounce off and move in the opposite direction. This property is useful for building single-Goomba chambers: if there is a one-way \ICON{smm/one-way}\ into a chamber with the width of a single Goomba that can only be accessed from the side, at most one Goomba can occupy it. This is because, if any other Goombas attempt to walk into the space, they will bounce off and be unable to pass through the one-way.

    \begin{figure}[t]
        \centering
        \includegraphics[width=1\linewidth]{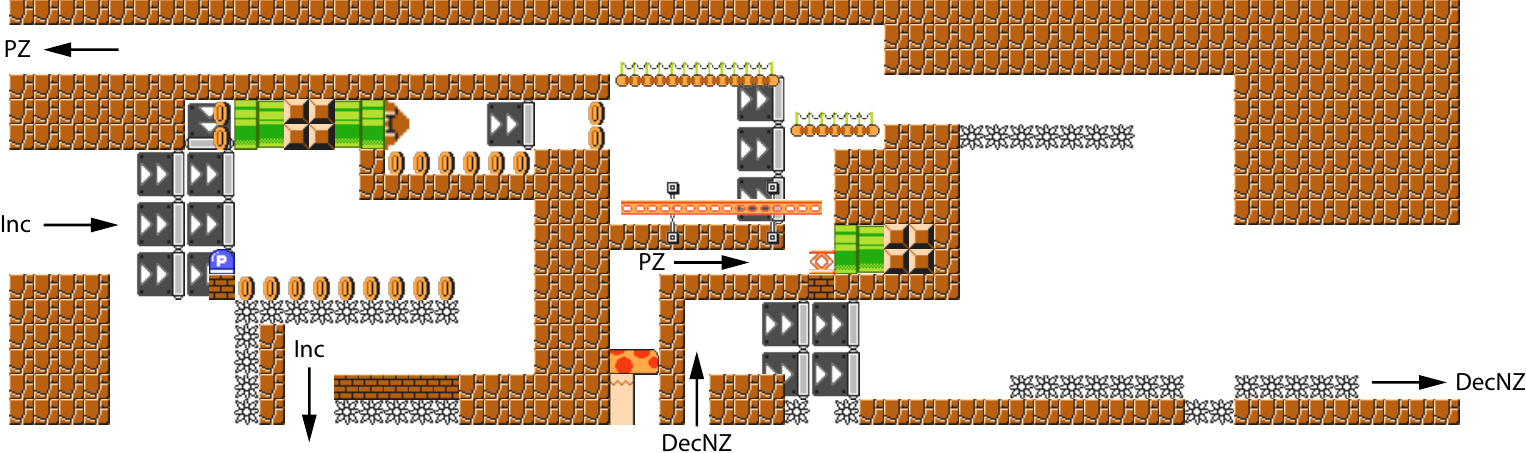}
        \caption{The counter gadget for Super Mario Maker.}
        \label{fig:smm}
    \end{figure}

\paragraph{Inc.}

    On traversal of the increment path, the player is forced to add exactly one Goomba \ICON{smm/goomba}\ to the gadget counter. Above the semisolid platform \ICON{smm/platform}\ counter, there is a long horizontal chamber with coins \ICON{smm/coin}\ for a floor. At one end of the corridor, a pipe \ICON{smm/pipe}\ spawns Goombas which immediately fall through the coin floor. At the other end of the pipe, a one-way \ICON{smm/one-way}\ leads to a space wide enough for exactly one Goomba bounded by coins: a single-Goomba chamber as described above. This chamber is level with the pipe and has a solid floor, unlike the coin floor of the corridor.

    On activation of a P-switch \ICON{smm/p-switch}, the coins \ICON{smm/coin}\ turn to brick blocks \ICON{smm/brick}. Any Goombas \ICON{smm/goomba}\ already in the gadget are defeated instantly, and future spawned Goombas can freely walk towards the chamber. The coins bordering the chamber on the other side will be brick blocks, trapping exactly one Goomba in the chamber. Once the P-switch timer expires, the floor will transform back into coins and all the Goombas currently in the corridor will be too low to enter the chamber. The brick blocks at the end of the chamber will turn back into coins, allowing the Goomba to pass through and fall onto the semisolid platforms \ICON{smm/platform}\ to increment the counter. As long as the corridor length and pipe spawn frequency guarantee that at least one Goomba reaches the chamber by the end of a P-switch cycle, every P-switch activation guarantees a counter increment by exactly one.

    An Inc traversal path can be built such that Mario is forced to hit a P-switch exactly once. We accomplish this by spawning P-switches \ICON{smm/p-switch}\ from a pipe \ICON{smm/pipe}\ in the ceiling, which fall onto a brick block \ICON{smm/brick}\ behind a one-way \ICON{smm/one-way}. Beyond the brick block, a row of coins \ICON{smm/coin}\ above a row of spikes \ICON{smm/spike}\ leads to solid ground \ICON{smm/ground}. After the solid ground, another row of spikes \ICON{smm/spike}\ below a row of brick blocks \ICON{smm/brick}\ lead out to the exit port. The gadget is constructed such that Mario must jump through a set of one-ways \ICON{smm/one-way}: one before hitting the P-switch \ICON{smm/p-switch}\ and another immediately after.
    
    To traverse, Mario must jump past the one-way \ICON{smm/one-way}\ and land on the P-switch \ICON{smm/p-switch}, then run across the transformed brick blocks \ICON{smm/brick}\ over the spikes \ICON{smm/spike}\ and onto the solid ground \ICON{smm/ground}. Once through the one-way, Mario cannot go back and activate another P-switch. Since the exit port is blocked by a spike row, Mario is forced to wait out the full P-switch timer, guaranteeing that the Goomba \ICON{smm/goomba}\ remains loaded until it can fall onto the counter and reach global ground.

\paragraph{DecNZ.}

    At the end of the semisolid counter platform \ICON{smm/platform}\ farthest from the chamber described above, another Goomba chamber is formed by a one-way \ICON{smm/one-way}\ and solid ground \ICON{smm/ground}\ tiles. This guarantees that, in any state $n>0$, exactly one Goomba \ICON{smm/goomba}\ is in the chambered section of the counter. Below the chamber, a pipe \ICON{smm/pipe}\ spawns trampolines \ICON{smm/spring}\ onto a brick block \ICON{smm/brick}\ which is accessible from below. Directly above the chamber, a semisolid platform \ICON{smm/semisolid}\ leads toward a long fall. The platform is low enough such that a Goomba bounced up from the counter by a trampoline would land on the upper semisolid platform \ICON{smm/semisolid}, no longer in the counter. The upper semisolid platform is bounded on one side by one-ways \ICON{smm/one-way}, and on the other it has a long fall onto a row of spikes \ICON{smm/spike}.

    The decrement works by having small Mario hit the brick block \ICON{smm/brick}\ from below, bouncing the trampoline \ICON{smm/spring}\ upwards and allowing it to bounce a Goomba \ICON{smm/goomba}\ up into the air onto the upper semisolid platform \ICON{smm/semisolid}. Mario must then traverse the row of spikes \ICON{smm/spike}, which is too long to clear in a single jump. If the counter had more than one Goomba and the bounce succeeded, a Goomba will fall onto the spikes, allowing Mario to jump off it to gain extra height and clear the row, exiting the gadget. If the counter was at zero, Mario would be unable to clear the spikes and lose a life instead. This checks that the decrement removes at least one Goomba.

    We enforce the condition that Mario can only decrement once per traversal by using one-ways. Mario is forced to jump through a one-way \ICON{smm/one-way}, hit the brick block \ICON{smm/brick}, then fall through another one-way \ICON{smm/one-way}. Using this construction, Mario can only hit the brick block once, making a tight lower bound and guaranteeing decrement of exactly one enemy.

\paragraph{PZ.}

    The traverse path requires that Mario jumps through the semisolid counter platform \ICON{smm/semisolid}\ from below, in particular through the bottom of the single-Goomba chamber. If the gadget is in any nonzero state, there is guaranteed to be one Goomba \ICON{smm/goomba}\ in the chamber. Mario will take damage when making contact from below and lose a life. If the gadget is in a zero state instead, the Goomba chamber will be empty and Mario will be free to traverse.
\end{proof}

\subsection{Super Mario Maker 2 (Normal Styles)}

\begin{theorem}
    All four normal game styles (Super Mario Bros.\ 1, Super Mario Bros.\ 3, Super Mario World, and New Super Mario Bros.\ U) of Super Mario Maker 2 are RE-complete.
\end{theorem}

The reduction for Super Mario Maker 1 from Section~\ref{sec:smm1} might be adaptable to Super Mario Maker 2, but differing mechanics in the latter game (which allows picking up items from behind one-ways) would mean that the gadget would have to be altered for some of the game styles in order to prevent breaking the gadget. Thus, we demonstrate here a different gadget which works equally well for all four styles, making use of some mechanics exclusive to Super Mario Maker 2.

\begin{proof}
    We reduce from planar reachability with the Inc-DecNZ-PZ gadget
    (Theorem~\ref{planar}).
    Figure \ref{smm2-ov} shows the gadget,
    which is largely made up of the following elements:

    \begin{figure}[ht]
    \centering
    \includegraphics[width=\textwidth]{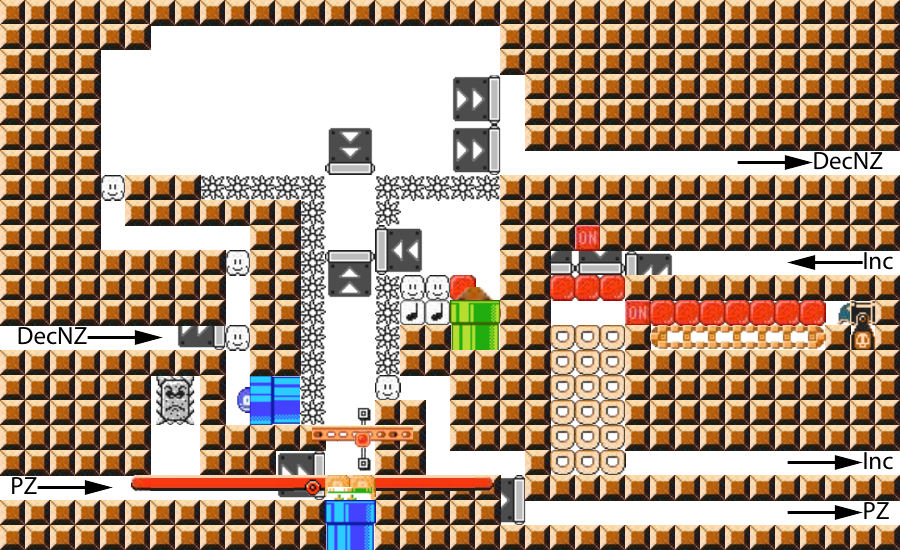}
    \caption{The complete Inc-DecNZ-PZ gadget for the normal game styles of Super Mario Maker 2.}
    \label{smm2-ov}
    \end{figure}

    \begin{itemize}
        \item \textbf{Solid ground:} \ICON{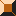}\ Standard blocks with no special ability.
        \item \textbf{Semisolid platforms/cloud blocks:} \ICON{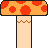}\ \ICON{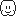}\ Platforms which can be jumped through from below but are solid from above.
        \item \textbf{Donut blocks:} \ICON{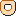}\ Semisolid blocks that begin to fall after Mario stands on them for a short amount of time. Placing many in a vertical drop effectively slows Mario's downward traversal.
        \item \textbf{Note blocks:} \ICON{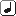}\ Bouncy blocks that cause any entity that walks on them to be bounced upward. 
        \item \textbf{One-ways:} \ICON{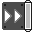}\ Walls which allow entities (Mario, enemies, clown cars, etc.)\ to pass the white bar only in the direction of the arrows.
        \item \textbf{P-switches:} \ICON{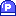}\ Switches which invert the state of P-blocks \ICON{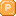}\ \ICON{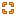}\ for the duration of a timer.
        \item \textbf{P-blocks:} \ICON{smm2/p-block1}\ \ICON{smm2/p-block2}\ These blocks flip from being outlines to being solid (or vice versa) while the P-switch \ICON{smm2/p-switch}\ is active. In this construction, they prevent the clown car \ICON{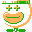}\ from spawning out of the blue pipe unless the P-switch is active.
        \item \textbf{On/off switches:} \ICON{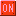}\ \ICON{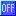}\ When hit from below, these switches flip a global on/off state, toggling the solidity of on/off blocks: \ICON{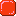}\ \ICON{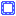}\ vs.\ \ICON{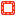}\ \ICON{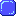}.
        \item \textbf{On/off blocks:} \ICON{smm2/red-on}\ \ICON{smm2/red-off}\ \ICON{smm2/blue-on}\ \ICON{smm2/blue-off}\ These blocks come in and out of existence based on the state of the on/off switches \ICON{smm2/switch-on}\ \ICON{smm2/switch-off}:
        \ICON{smm2/red-on}\ \ICON{smm2/blue-off}\ in the ``on'' red state \ICON{smm2/switch-on}, and \ICON{smm2/red-off}\ \ICON{smm2/blue-on}\ in the ``off'' blue state \ICON{smm2/switch-off}.
        \item \textbf{Tracks:} \ICON{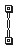}\ Rails specifying periodic movement of attached entities; platforms \ICON{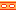}\ on tracks are global ground, meaning that they and entities on them do not despawn when offscreen.
        \item \textbf{Spikes:} \ICON{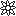}\ Blocks which damage Mario upon contact.
        \item \textbf{Blaster:} \ICON{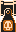}\ These blasters shoot a shell \ICON{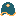}\ in the direction of Mario if he is sufficiently close and the corresponding barrel of the blaster is not blocked. These shells can activate on/off switches \ICON{smm2/switch-on}\ \ICON{smm2/switch-off}.
        \item \textbf{Clown cars:} \ICON{smm2/clown}\ These vehicles allow exactly one falling enemy to enter and ride them. If a loaded clown car touches spikes \ICON{smm2/spike}, then it will panic and fly upward.
        \item \textbf{Goombas:} \ICON{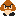}\ Enemies which damage Mario when he contacts from any direction but above. Can safely walk on spikes and are defeated by being jumped on, giving Mario a vertical boost similar to a jump.
        We use big ($2 \times 2$) Goombas in this construction.
        \item \textbf{Pipes:} \ICON{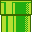}\ Elements which periodically a particular entity (drawn next to the pipe) into the course. If the pipe spawns items, it will only do so if the last element that the pipe spawned is no longer loaded or exists; for example, a pipe will only spawn more clown cars \ICON{smm2/clown}\ if the last clown car it got destroyed. Enemies such as Goombas \ICON{smm2/goomba}\ spawn indefinitely. If either block in front of the pipe is blocked, the pipe is prevented from spawning its object.
        \item \textbf{Thwomps:} \ICON{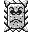}\ An enemy that charges downward if Mario is sufficiently close and at or below the level of the Thwomp, and then resets to its original location.
        \item \textbf{Seesaws:} \ICON{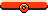}\ A tilting platform that balances based on the weight on both sides of the center. If a Thwomp \ICON{smm2/thwomp}\ pounds on one side of the seesaw, it will launch up any enemies on the other side of the platform. This is used to get the Goombas \ICON{smm2/goomba}\ into a falling state so that they can be picked up by the clown car \ICON{smm2/clown}.
        \item \textbf{Conveyors:} \ICON{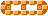}\ A moving platform that can slow down a shell \ICON{smm2/shell}\ that is moving in the opposite direction. In this construction, it allows us to use less space to achieve the timing for the shell that deactivates an activated on/off switch \ICON{smm2/switch-on}.
    \end{itemize}

    As before, the state of the counter gadget is the number of Goombas \ICON{smm2/goomba}\ on the platform \ICON{smm2/platform}\ on the track \ICON{smm2/track}.

    \begin{figure}[ht]
    \centering
    \begin{subfigure}[b]{0.49\linewidth}
        \centering
        \includegraphics[width=\linewidth]{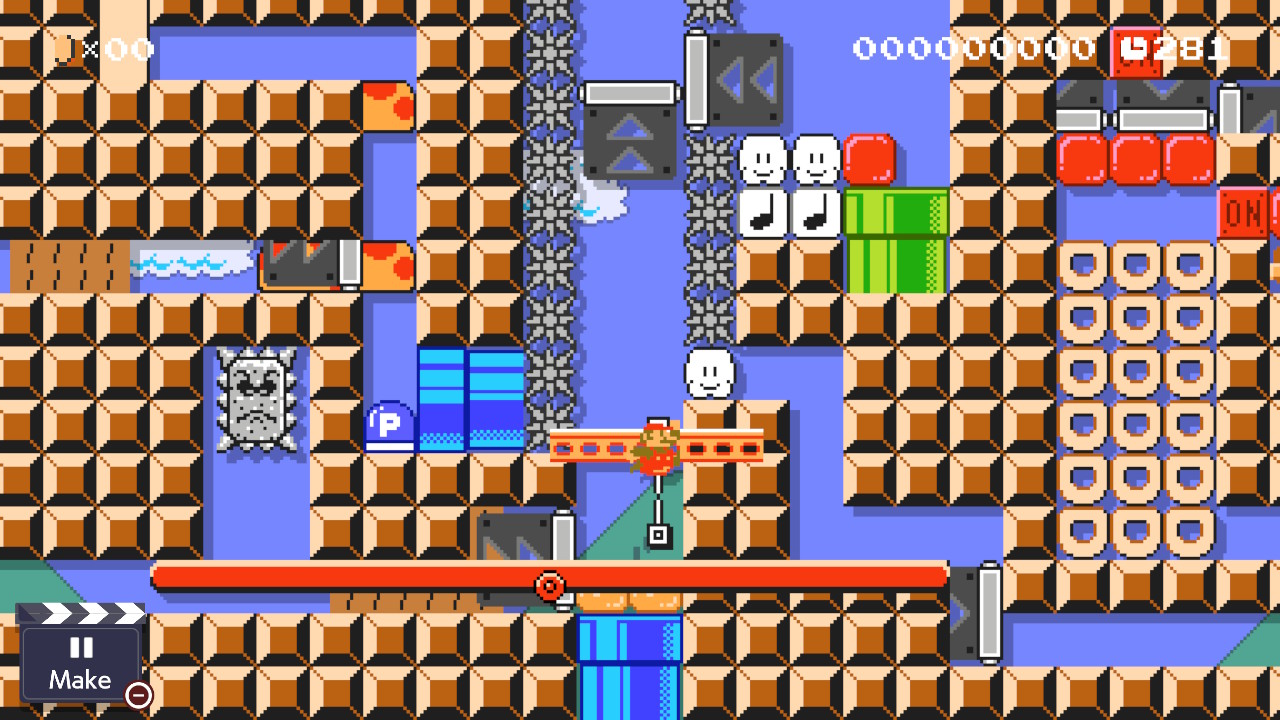}
        \caption{Mario can traverse through the PZ path if the counter = 0.}
        \label{smm2-pz-good}
    \end{subfigure}
    \begin{subfigure}[b]{0.49\linewidth}
        \centering
        \includegraphics[width=\linewidth]{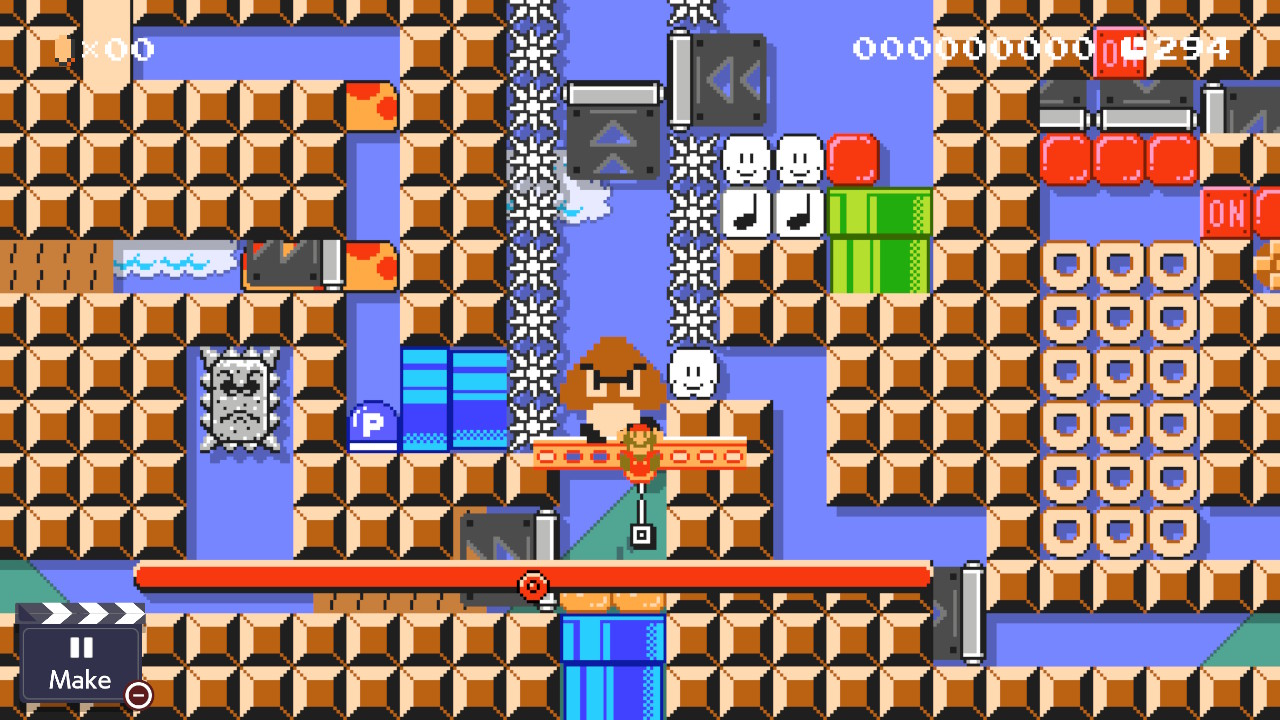}
        \caption{Mario dies traversing the PZ path if the counter $>$ 0.}
        \label{smm2-pz-bad}
    \end{subfigure}
    \caption{Screenshots of traversing the PZ path for the normal game styles of Super Mario Maker 2.}
    \label{smm2-pz}
    \end{figure}

    \paragraph{PZ.}
    When the gadget is in state 0, there are no Goombas \ICON{smm2/goomba}\ on the platform \ICON{smm2/platform}, allowing the player to traverse the PZ path (which starts from the bottom-left tunnel and ends in the bottom-right tunnel of Figure \ref{smm2-pz-good}). When the gadget is in any state $>0$, there is at least one Goomba on the platform, preventing the player from traversing the PZ path without taking damage (and thus dying, as in Figure \ref{smm2-pz-bad}).

    \begin{figure}[ht]
    \centering
    \begin{subfigure}[b]{0.49\linewidth}
        \centering
        \includegraphics[width=\linewidth]{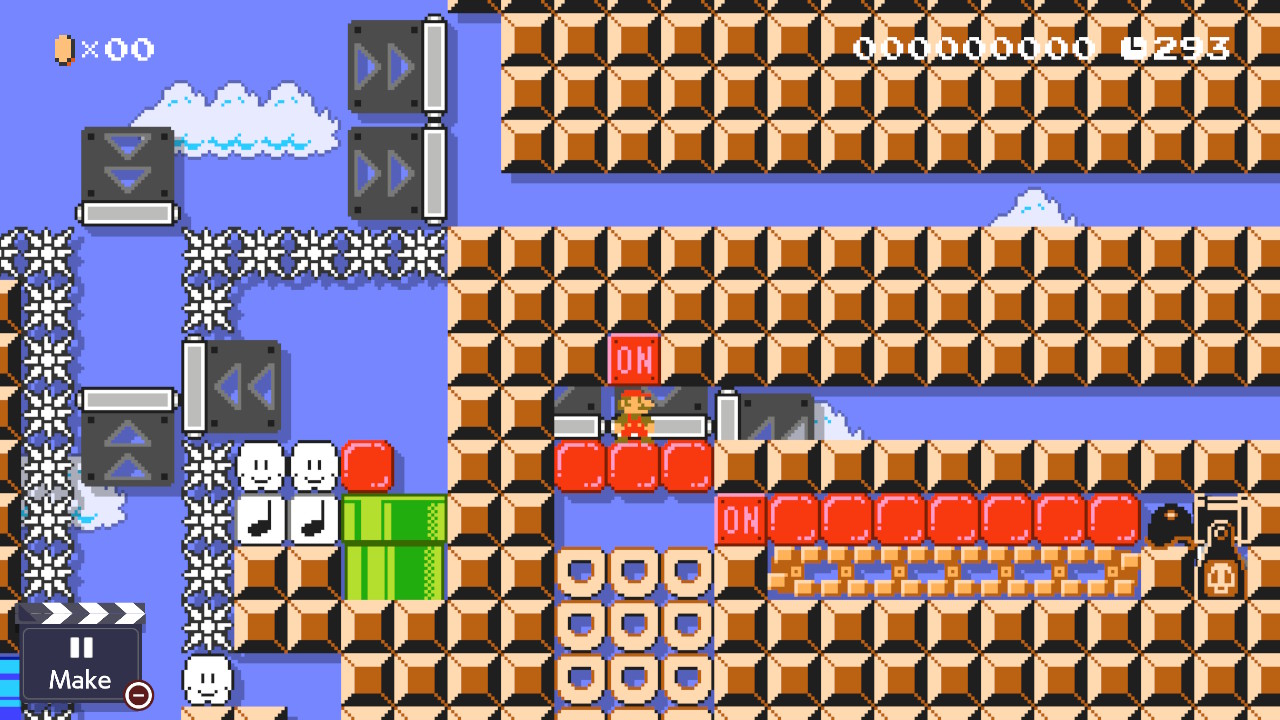}
        \caption{Mario hits the on/off block \ICON{smm2/switch-on}.\newline}
        \label{smm2-inc-1}
    \end{subfigure}
    \begin{subfigure}[b]{0.49\linewidth}
        \centering
        \includegraphics[width=\linewidth]{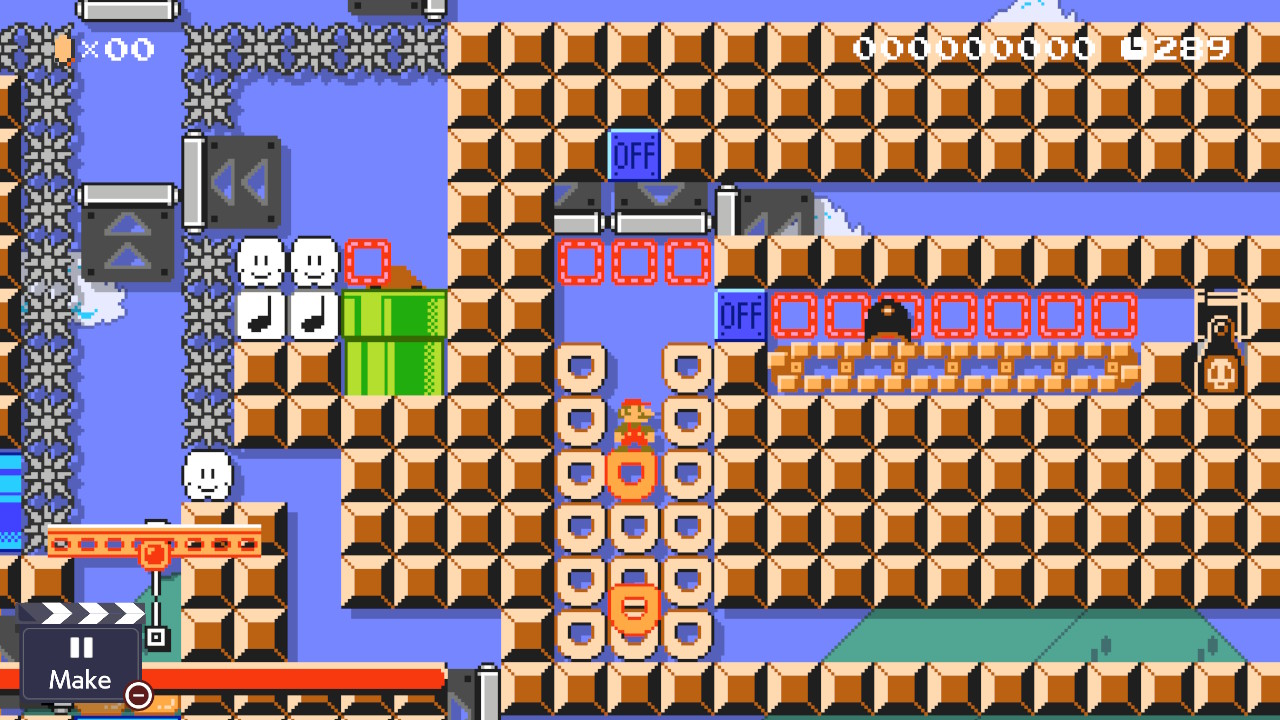}
        \caption{The shell \ICON{smm2/shell}\ travels left as the now-unblocked pipe \ICON{smm2/pipe}\ begins spawning Goombas \ICON{smm2/goomba}.}
        \label{smm2-inc-2}
    \end{subfigure}
    \begin{subfigure}[b]{0.49\linewidth}
        \centering
        \includegraphics[width=\linewidth]{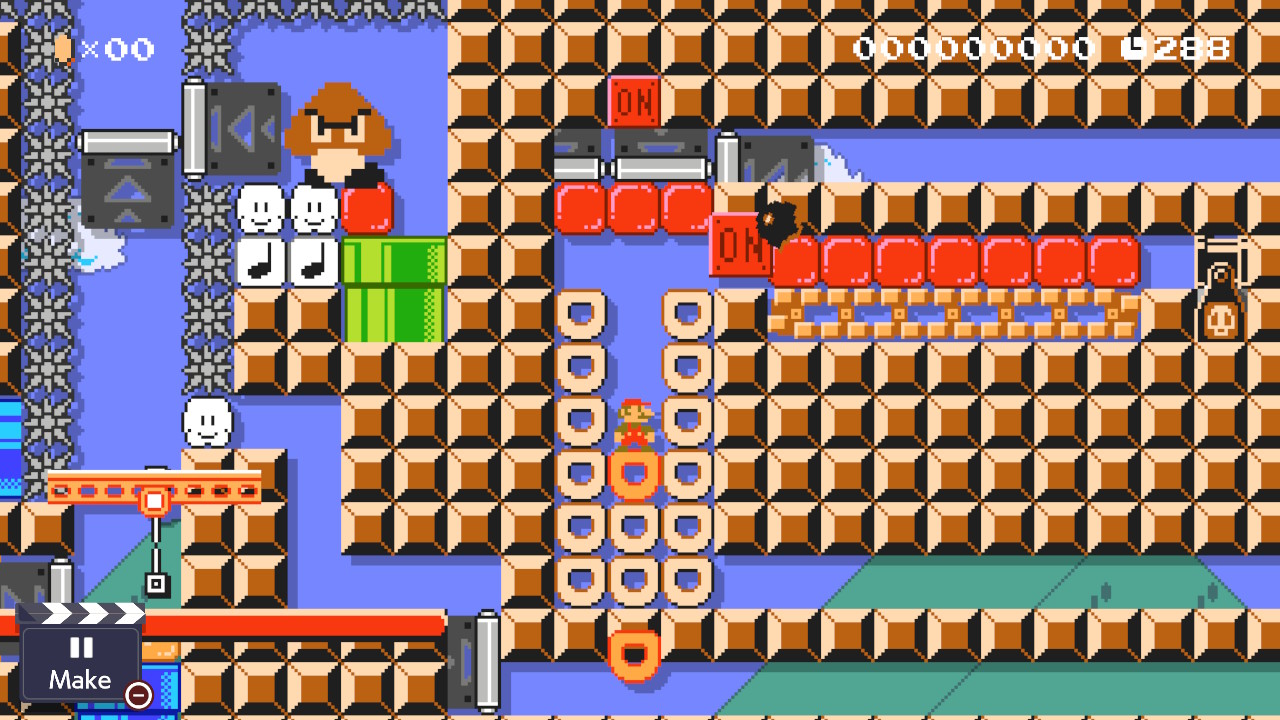}
        \caption{The shell \ICON{smm2/shell}\ resets the on/off switch \ICON{smm2/switch-off}\ after exactly one Goomba \ICON{smm2/goomba}\ has spawned.}
        \label{smm2-inc-3}
    \end{subfigure}
    \begin{subfigure}[b]{0.49\linewidth}
        \centering
        \includegraphics[width=\linewidth]{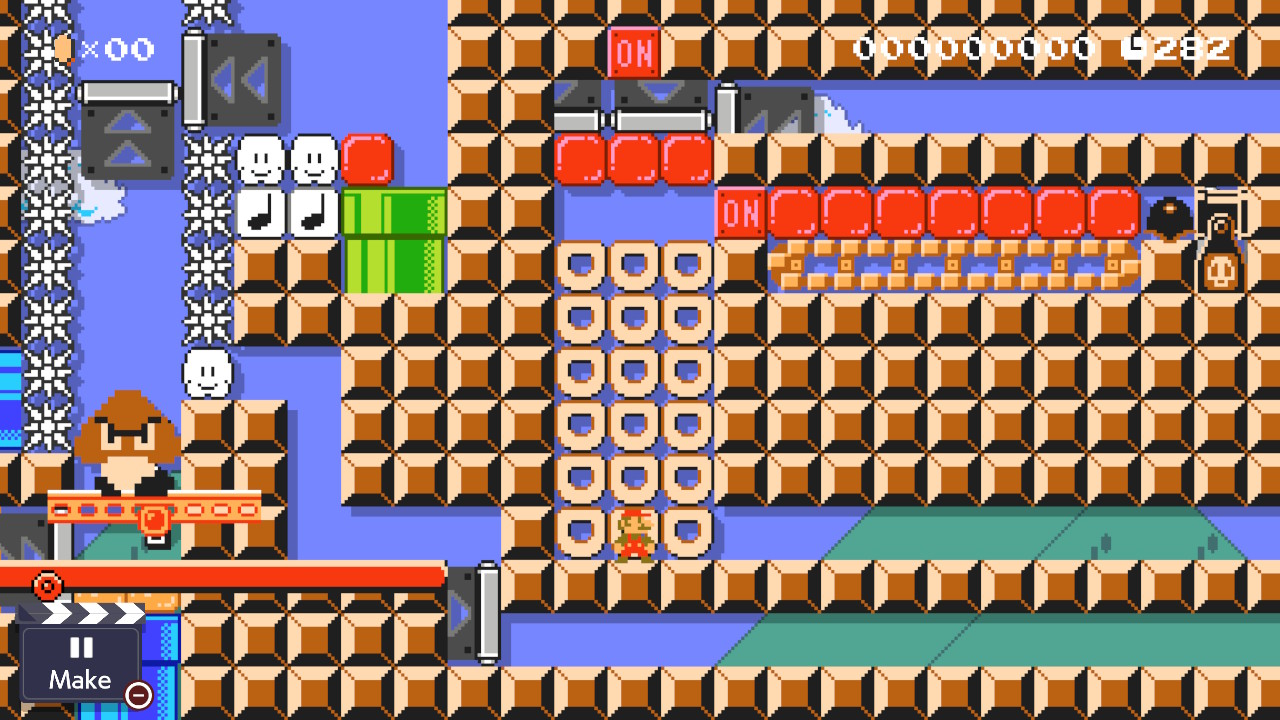}
        \caption{Mario is forced to stay loading the Goomba \ICON{smm2/goomba}\ until it has reached global ground.}
        \label{smm2-inc-4}
    \end{subfigure}
    \caption{Screenshots of traversing the Inc path for the normal game styles of Super Mario Maker 2.}
    \label{smm2-inc}
    \end{figure}

\paragraph{Inc.}
    When the player enters the increment path (whose entrance is the top-left tunnel and whose exit is the bottom-right tunnel in Figure \ref{smm2-inc-2}), the first tunnel is long enough that a shell \ICON{smm2/shell}\ will be launched left from the blaster \ICON{smm2/bill}\ and bounce within the single adjacent space. At the end of this tunnel, the leftward one-way \ICON{smm2/one-way}\ ensures that the player cannot hit the on/off switch \ICON{smm2/switch-on}\ without fully activating the increment, and the downward facing one-ways and on/off blocks below ensure that the player cannot hit the on/off switch more than once (all as shown in Figure \ref{smm2-inc-1}). 
    
    Once the on/off switch \ICON{smm2/switch-on}\ is triggered, the green pipe \ICON{smm2/pipe}, being no longer blocked from spawning, immediately starts spawning Goombas \ICON{smm2/goomba}. At the same time, the shell \ICON{smm2/shell}, no longer confined to a single space, starts moving left toward an on/off switch \ICON{smm2/switch-off}\ (both of which are shown in Figure \ref{smm2-inc-2}). The shell's distance to the switch, along with the speed of the conveyor \ICON{smm2/conveyor}\ and the spawning speed of the green pipe \ICON{smm2/pipe}, ensures that exactly one Goomba \ICON{smm2/goomba}\ can spawn before the pipe is again blocked from spawning by the shell hitting the on/off switch (as shown in Figure \ref{smm2-inc-3}). This single Goomba \ICON{smm2/goomba}\ moves left, bouncing up from walking on note blocks \ICON{smm2/note}, and falls onto the platform \ICON{smm2/platform}\ on the track \ICON{smm2/track}.
    
    The rows of donut blocks \ICON{smm2/donut}\ that Mario must endure at the end of this path ensures that the player keeps the Goomba \ICON{smm2/goomba}\ loaded on screen until it is on the global ground of the platform \ICON{smm2/platform}\ (as shown in Figure \ref{smm2-inc-4}), so that the counter value is correct.

    \begin{figure}[ht]
    \centering
    \begin{subfigure}[b]{0.49\linewidth}
        \centering
        \includegraphics[width=\linewidth]{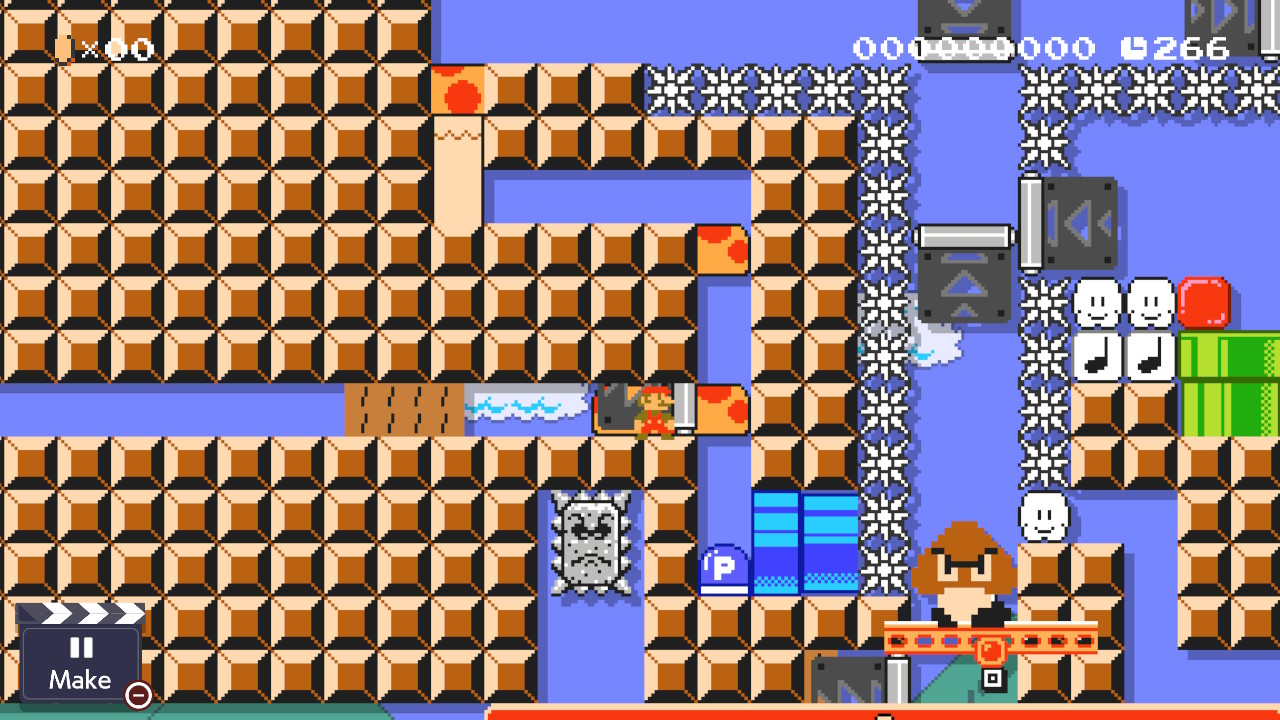}
        \caption{Mario falls onto the P-switch \ICON{smm2/p-switch}.\newline}
        \label{smm2-dec-1}
    \end{subfigure}
    \begin{subfigure}[b]{0.49\linewidth}
        \centering
        \includegraphics[width=\linewidth]{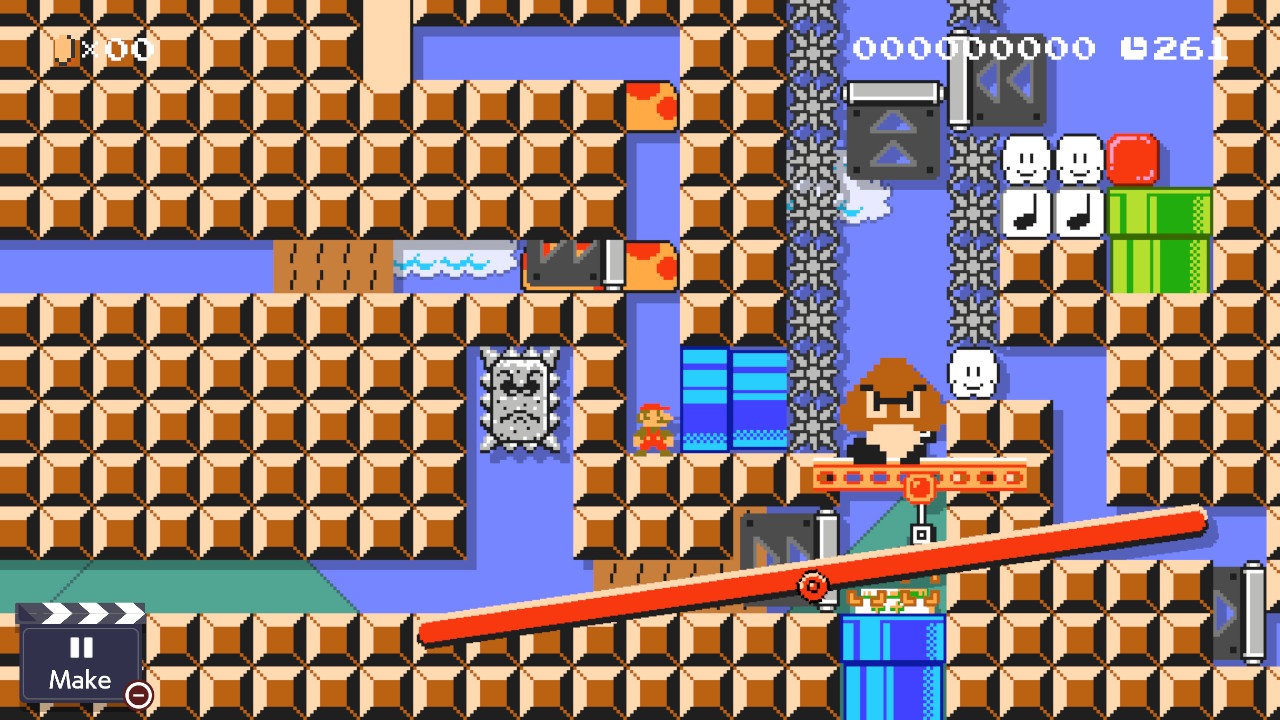}
        \caption{The unblocked pipe \ICON{smm2/pipe}\ spawns a clown car \ICON{smm2/clown}; the Thwomp \ICON{smm2/thwomp}\ repeatedly charges downward.}
        \label{smm2-dec-2}
    \end{subfigure}

    \smallskip
    \begin{subfigure}[b]{0.49\linewidth}
        \centering
        \includegraphics[width=\linewidth]{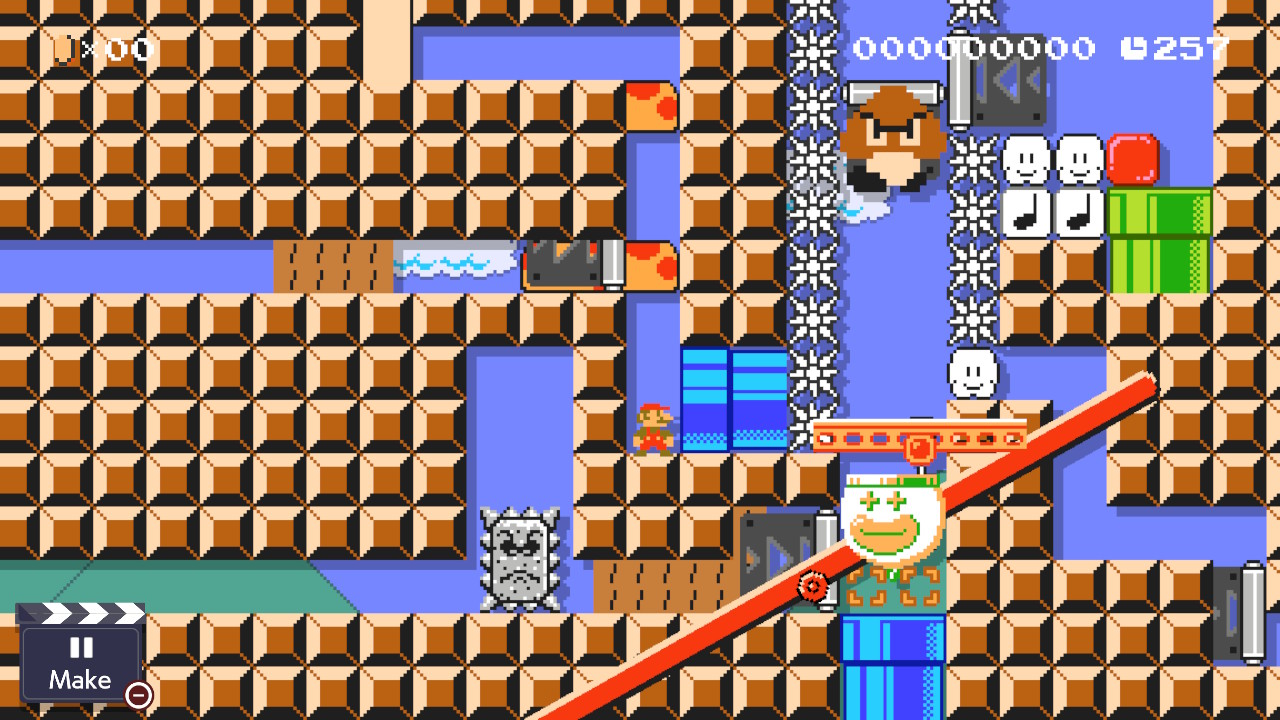}
        \caption{The charging Thwomp \ICON{smm2/thwomp}\ hits one side of the seesaw \ICON{smm2/seesaw}, launching the Goombas \ICON{smm2/goomba}\ upward.}
        \label{smm2-dec-3}
    \end{subfigure}
    \begin{subfigure}[b]{0.49\linewidth}
        \centering
        \includegraphics[width=\linewidth]{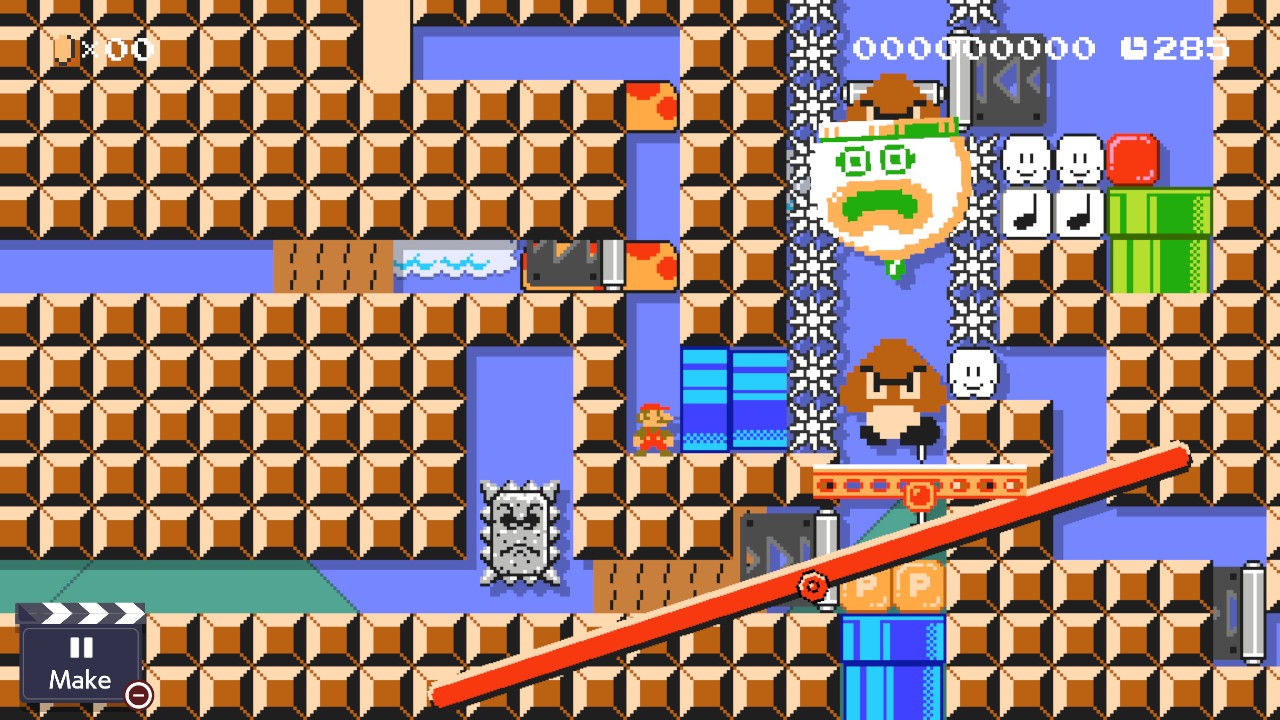}
        \caption{A single Goomba \ICON{smm2/goomba}\ rides the panicked clown car \ICON{smm2/clown}\ upward.}
        \label{smm2-dec-4}
    \end{subfigure}

    \smallskip
    \begin{subfigure}[b]{0.49\linewidth}
        \centering
        \includegraphics[width=\linewidth]{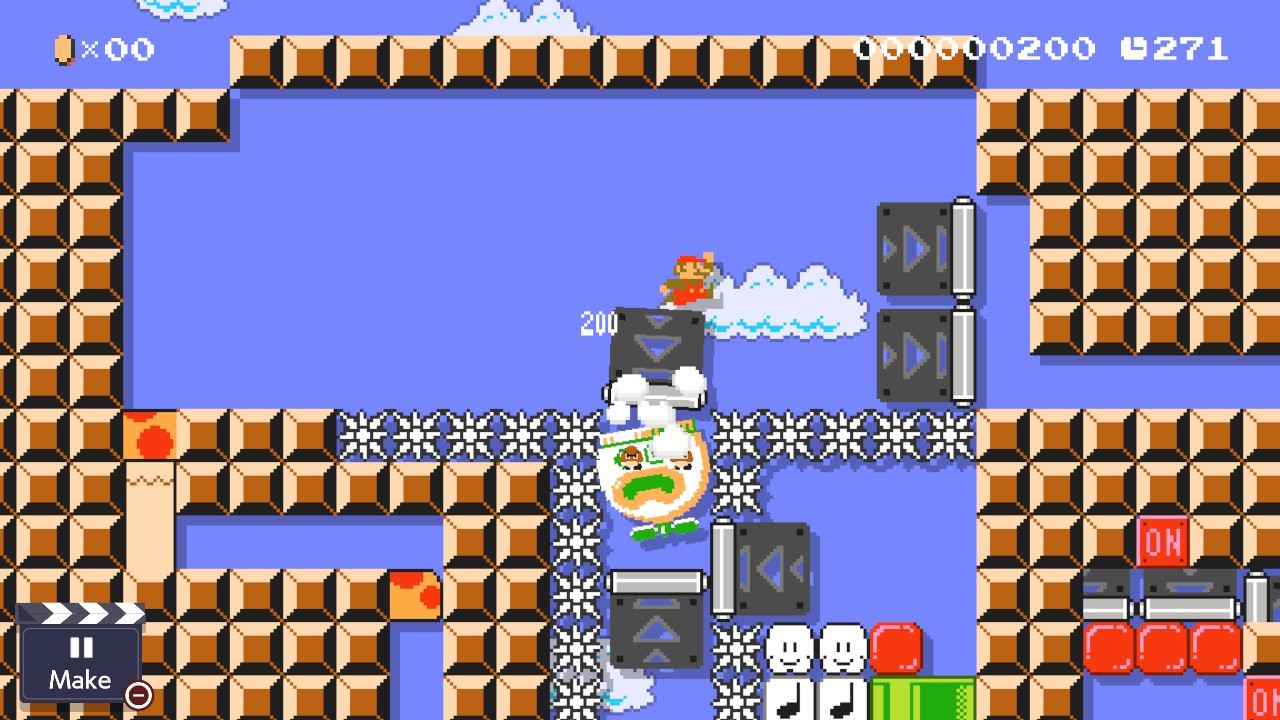}
        \caption{Mario needs to bounce on the Goomba \ICON{smm2/goomba}\ in the clown car \ICON{smm2/clown}\ to jump over the spikes \ICON{smm2/spike}.}
        \label{smm2-dec-5}
    \end{subfigure}
    \caption{Screenshots of traversing the DecNZ path for the normal game styles of Super Mario Maker 2.}
    \label{smm2-dec}
    \end{figure}

\paragraph{DecNZ.}
    Meanwhile, when the player enters the decrement path (whose entrance is the center-left tunnel in Figure \ref{smm2-dec-1} and whose exit is the center-right tunnel in Figure \ref{smm2-dec-5}), they fall onto a P-switch \ICON{smm2/p-switch}\ spawned from the the left blue pipe \ICON{smm2/pipe}, activating the P-switch, as shown in Figure \ref{smm2-dec-1}. (Note that Mario now blocks spawns from the pipe, so no new P-switches will spawn. And even if the player managed to fall into the hole exactly when the P-switch was spawning, since it is a 1-wide hole, they can never pick up the P-switch, and only fall onto it or block spawns from the pipe.)
    
    This causes the bottom blue pipe \ICON{smm2/pipe}, no longer being blocked from spawning, to spawn a clown car \ICON{smm2/clown}\ (as shown in Figure \ref{smm2-dec-2}). At the same time, falling into this hole is just far enough down to trigger the Thwomp \ICON{smm2/thwomp}, which repeatedly charges downward onto the seesaw \ICON{smm2/seesaw}. This, in turn, causes the opposite end of the seesaw to shoot up, launching the Goombas \ICON{smm2/goomba}\ (both of which are shown in Figure \ref{smm2-dec-3}). 
    
    Though it usually takes a few seconds for all of this to line up, eventually this results in the Goombas \ICON{smm2/goomba}\ being shot up while the clown car \ICON{smm2/clown}\ is under them, causing one Goomba to enter the clown car while the rest fall back onto the platform \ICON{smm2/platform}. Now that it is entered, the clown car panics from touching spikes \ICON{smm2/spike}, and flies upward until it gets stuck between the two one-ways \ICON{smm2/one-way}\ at the top of the spike column.
    Figure~\ref{smm2-dec-4} shows the separation of a single Goomba and panicking of the clown car.
    
    Now the player is free to jump up to the top tunnel of the decrement path, where they can clear the long jump over the spikes \ICON{smm2/spike}\ by bouncing on the Goomba \ICON{smm2/goomba}\ in the clown car \ICON{smm2/clown}\ in the middle of the jump (as shown in Figure \ref{smm2-dec-5}). If the player performs an illegal decrement from state 0, then the clown car will never fly up because it never gets entered, and Mario will not be able to clear the long jump over the spikes.
\end{proof}

\FloatBarrier
\subsection{Super Mario Maker 2 (3D World Style)}

The 3D world style of Super Mario Maker 2 has many different mechanics from that of the four regular styles, and so requires a different construction to show that it is RE-hard.

\begin{theorem}
    The Super Mario 3D World style of Super Mario Maker 2 is RE-complete.
\end{theorem}
\begin{proof}
    We reduce from reachability with the Inc-DecNZ-PZ gadget.
    Figure \ref{smm2-3dw-ov} shows the gadget,
    which is largely made up of the following elements:

    \begin{figure}[ht]
    \centering
    \includegraphics[width=\textwidth]{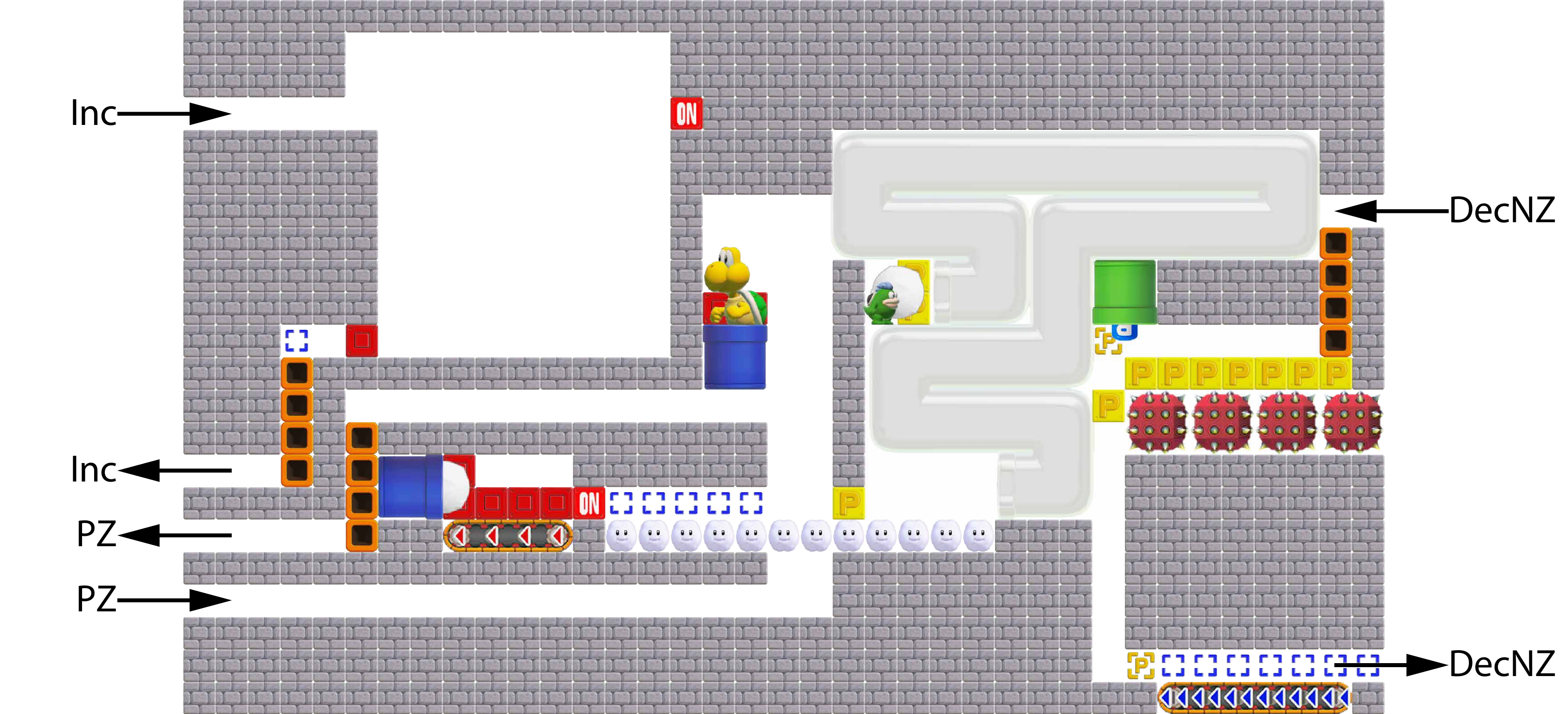}
    \caption{The complete Inc-DecNZ-PZ gadget for the 3D World game style of Super Mario Maker 2}
    \label{smm2-3dw-ov}
    \end{figure}

    \begin{itemize}
        \item \textbf{Solid ground:} \ICON{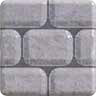}\ Standard blocks with no special effects.
        \item \textbf{Cloud blocks:} \ICON{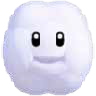}\ Platforms which can be jumped through from below but are solid from above. In the 3D World style only, cloud blocks are global ground, which means that which once spawned, they do not despawn and confer this property to enemies that are standing on them, allowing us to preserve the value of the counter even after it moves offscreen. The number of Koopas \ICON{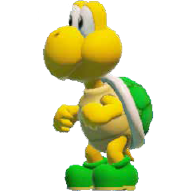}\ on the cloud blocks represents the state of the counter.
        \item \textbf{Donut blocks:} \ICON{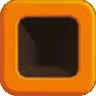}\ Semisolid blocks that begin to fall after Mario stands on them for a short amount of time. Placing many in a vertical drop effectively slows Mario's downward traversal. Because they are solid, placing four in a column enforces a one-way constraint, since by the time Mario has dropped the last donut block, the first will have respawned.
        \item \textbf{P-switches:} \ICON{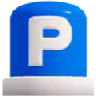}\ Switches which invert the state of P-blocks \ICON{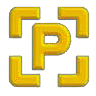}\ \ICON{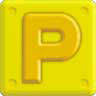}\ for the duration of a timer.
        \item \textbf{P-blocks:} \ICON{smm3d/p-block1}\ \ICON{smm3d/p-block3d}\ These blocks flip from being outlines to being solid (or vice versa) while the P-switch \ICON{smm3d/p-switch}\ is active.
        \item \textbf{On/off switches:} \ICON{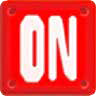}\ \ICON{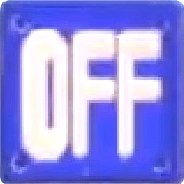}\ When hit from below, these switches flip a global on/off state, toggling the solidity of on/off blocks: \ICON{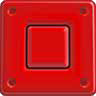}\ \ICON{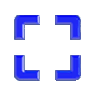}\ vs.\ \ICON{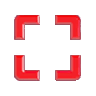}\ \ICON{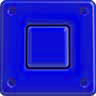}.
        \item \textbf{On/off blocks:} \ICON{smm3d/red-on}\ \ICON{smm3d/red-off}\ \ICON{smm3d/blue-on}\ \ICON{smm3d/blue-off}\ These blocks come in and out of existence based on the state of the on/off switches \ICON{smm3d/switch-on}\ \ICON{smm3d/switch-off}:
        \ICON{smm3d/red-on}\ \ICON{smm3d/blue-off}\ in the ``on'' red state \ICON{smm3d/switch-on}, and \ICON{smm3d/red-off}\ \ICON{smm3d/blue-on}\ in the ``off'' blue state \ICON{smm3d/switch-off}.
        \item \textbf{Red spike blocks:} \ICON{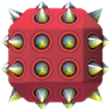}\ \ICON{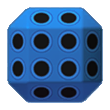}\ Solid terrain blocks with retractable spikes. When the on/off switch is in the ``on'' red state \ICON{smm3d/switch-on}, these blocks extend their spikes \ICON{smm3d/spike-on}, which damage Mario if he comes into contact with them.
        \item \textbf{Koopas:} \ICON{smm3d/koopa}\ Enemies which damage Mario when he contacts from any direction but above. The number of Koopas on the cloud blocks is the state of the counter.
        \item \textbf{Spike:} \ICON{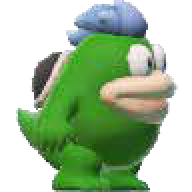}\ An enemy that continually throws snowballs \ICON{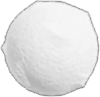}\ on a global timer. Rolling snowballs are destroyed if they collide with something else, like a P-block \ICON{smm3d/p-block3d}\ or on/off switch \ICON{smm3d/switch-on}\ \ICON{smm3d/switch-off}. They will also trigger the on/off switch if they roll into it. Finally, if they collide with a group of enemies such as Koopas \ICON{smm3d/koopa}, a rolling snowball will kill exactly one enemy before being destroyed.
        \item \textbf{Pipes:} \ICON{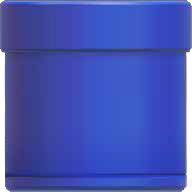}\ Elements which periodically a particular entity (drawn next to the pipe) into the course. If the pipe spawns items such as a snowball \ICON{smm3d/snowball}\ or P-switch \ICON{smm3d/p-switch}, it will only do so if the last element that the pipe spawned is no longer loaded or exists. Enemies such as Koopa \ICON{smm3d/koopa}\ spawn indefinitely. If either block in front of the pipe is blocked, the pipe is prevented from spawning its object. Different colors of pipes spawn objects at different rates; the specific pipe colors used in this construction were chosen for timing purposes. In this construction, the left blue pipe spawns snowballs \ICON{smm3d/snowball}, the center blue pipe spawns Koopas \ICON{smm3d/koopa}, and the green pipe spawns P-switches \ICON{smm3d/p-switch}.
        \item \textbf{Clear pipes:} \ICON{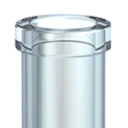}\ Tubes that allow for precise item transport. By throwing an object, such as a snowball \ICON{smm3d/snowball}, into one end of the pipe, the object travels the length of the clear pipe at a fixed speed and comes out of the other end as if still thrown.
        \item \textbf{Conveyors:} \ICON{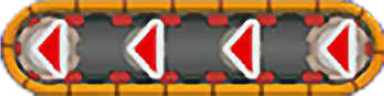}\ A moving platform that can slow down a snowball \ICON{smm3d/snowball}\ that is moving in the opposite direction. In this construction, it allows us to use less space to achieve the timing for the snowball that deactivates an activated on/off switch \ICON{smm3d/switch-on}.
    \end{itemize}

    \begin{figure}[ht]
    \centering
    \begin{subfigure}[b]{0.49\linewidth}
        \centering
        \includegraphics[width=\linewidth]{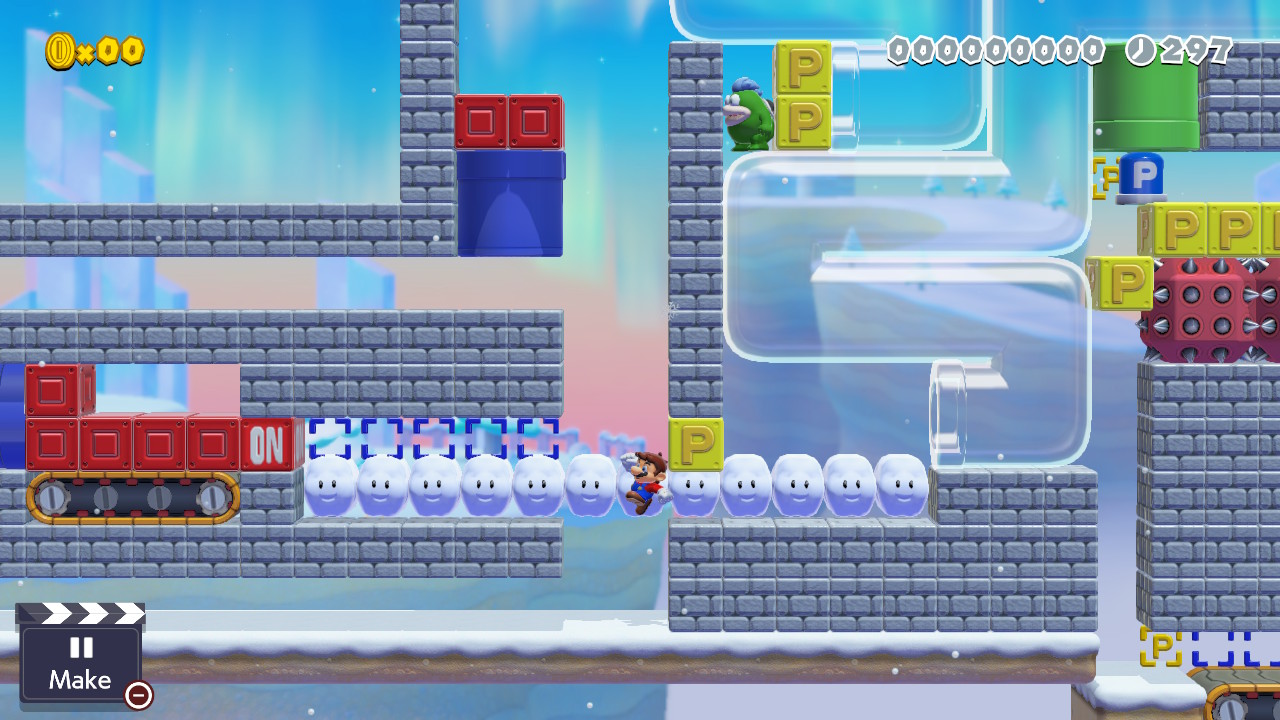}
        \caption{Mario can traverse through the PZ gadget if the counter = 0.}
        \label{smm2-3dw-pz-good}
    \end{subfigure}
    \begin{subfigure}[b]{0.49\linewidth}
        \centering
        \includegraphics[width=\linewidth]{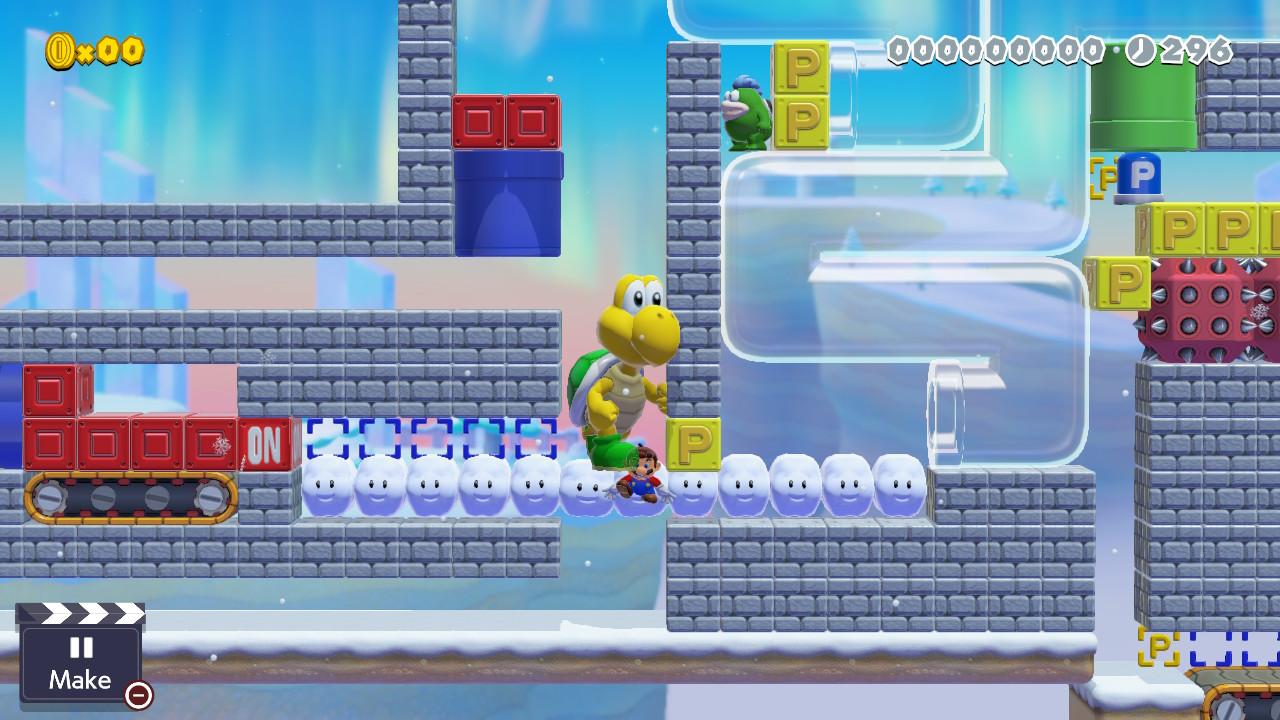}
        \caption{Mario dies traversing the PZ gadget if the counter $>$ 0.}
        \label{smm2-3dw-pz-bad}
    \end{subfigure}
    \caption{Screenshots of traversing the PZ gadget for the 3D World game style of Super Mario Maker 2.}
    \label{smm2-3dw-pz}
    \end{figure}
    \paragraph{Traverse.}
    When the gadget is in state 0, there are no Koopas \ICON{smm3d/koopa}\ on the cloud blocks \ICON{smm3d/cloud}, allowing Mario to traverse the PZ path (whose entrance is the bottom-left tunnel and whose exit is the middle-left tunnel in Figure \ref{smm2-3dw-pz-good}). Note that Mario can also enter the area near the middle blue pipe, but this accomplishes nothing for the player. 
    
    In the mechanics of 3D World, Mario can roll into the on/off switch \ICON{smm3d/switch-on}\ at the left end of the row of cloud blocks, trapping Mario in on/off blocks \ICON{smm3d/blue-on}\ and allowing the middle blue pipe to start spawning a Koopa \ICON{smm3d/koopa}\ (as explained later, only a single Koopa spawns because the on/off switch is triggered again by the snowball \ICON{smm3d/snowball}\ from the left blue pipe). Mario can then roll into the on/off switch \ICON{smm3d/switch-off}\ again to become free (and thereby kill the snowball from the left blue pipe), or wait for the snowball to turn off the on/off switch. Either Mario triggers the switch again in time to kill the spawning Koopa and thus does not affect the counter, or they trigger the switch too late (or allow it to be reversed by the snowball) in which case the Koopa \ICON{smm3d/koopa}\ falls onto cloud blocks \ICON{smm3d/cloud}\ and traps Mario before he can escape. Also, when the gadget is in any state $>0$, at least one Koopa on the cloud blocks similarly prevents Mario from traversing the PZ path without taking damage and thus dying, as in Figure \ref{smm2-3dw-pz-bad}.

    \begin{figure}[ht]
    \centering
    \begin{subfigure}{0.49\linewidth}
        \centering
        \includegraphics[width=\linewidth]{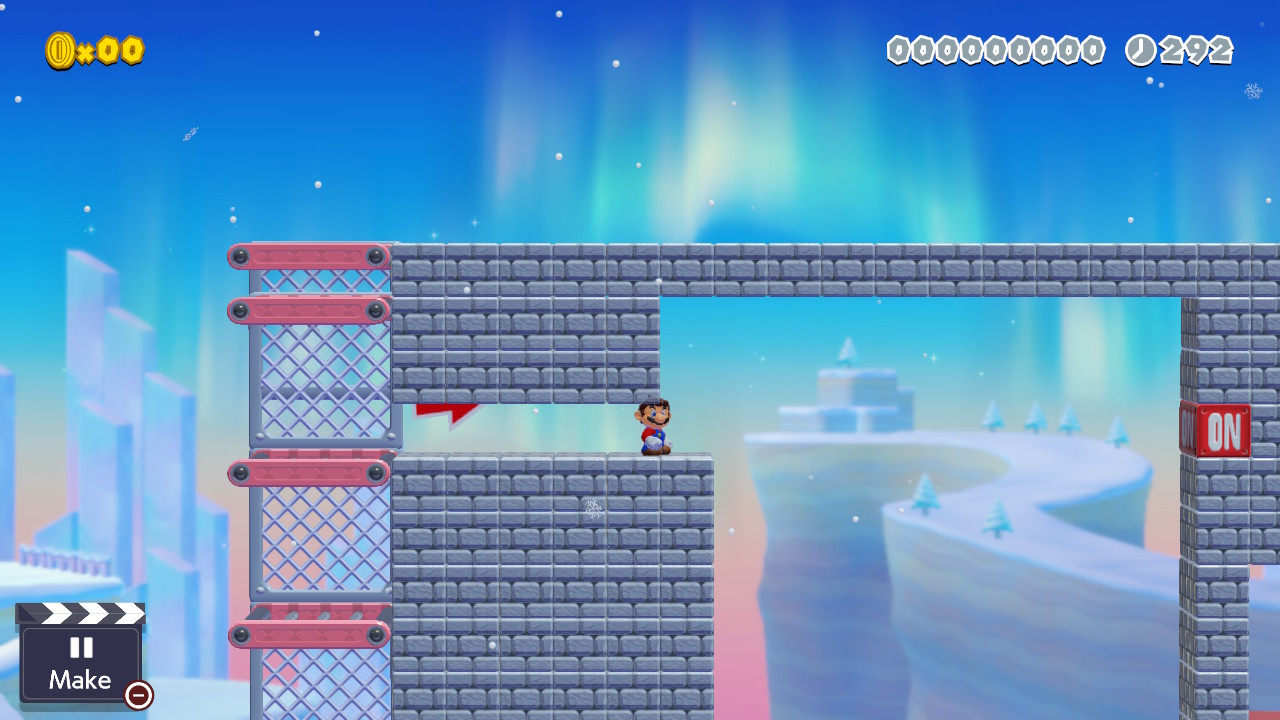}
        \caption{Mario long jumps toward the on/off block \ICON{smm3d/switch-on}.\newline}
        \label{smm2-3dw-inc-1}
    \end{subfigure}\hfill
    \begin{subfigure}{0.49\linewidth}
        \centering
        \includegraphics[width=\linewidth]{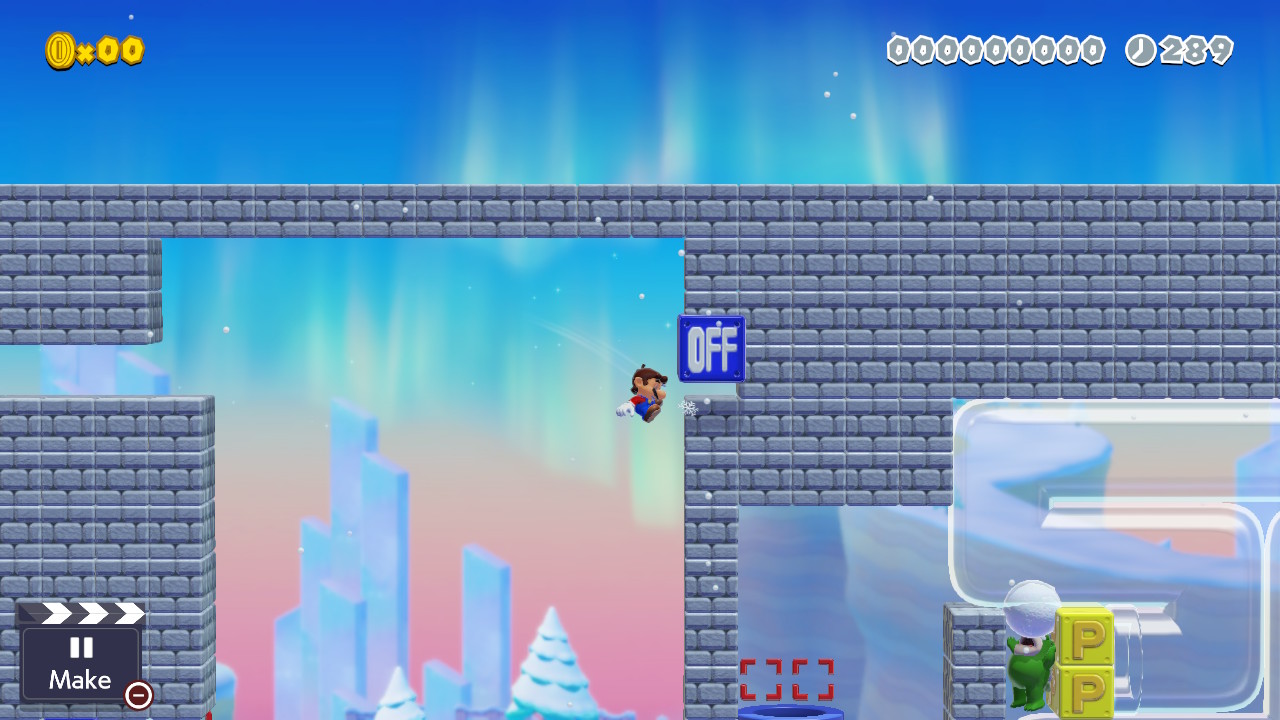}
        \caption{Mario toggles the on/off block \ICON{smm3d/switch-off}\ and the now unblocked right blue pipe \ICON{smm3d/pipe}\ begins spawning Koopas \ICON{smm3d/koopa}.}
        \label{smm2-3dw-inc-2}
    \end{subfigure}

    \smallskip
    \begin{subfigure}[t]{0.49\linewidth}
        \centering
        \includegraphics[width=\linewidth]{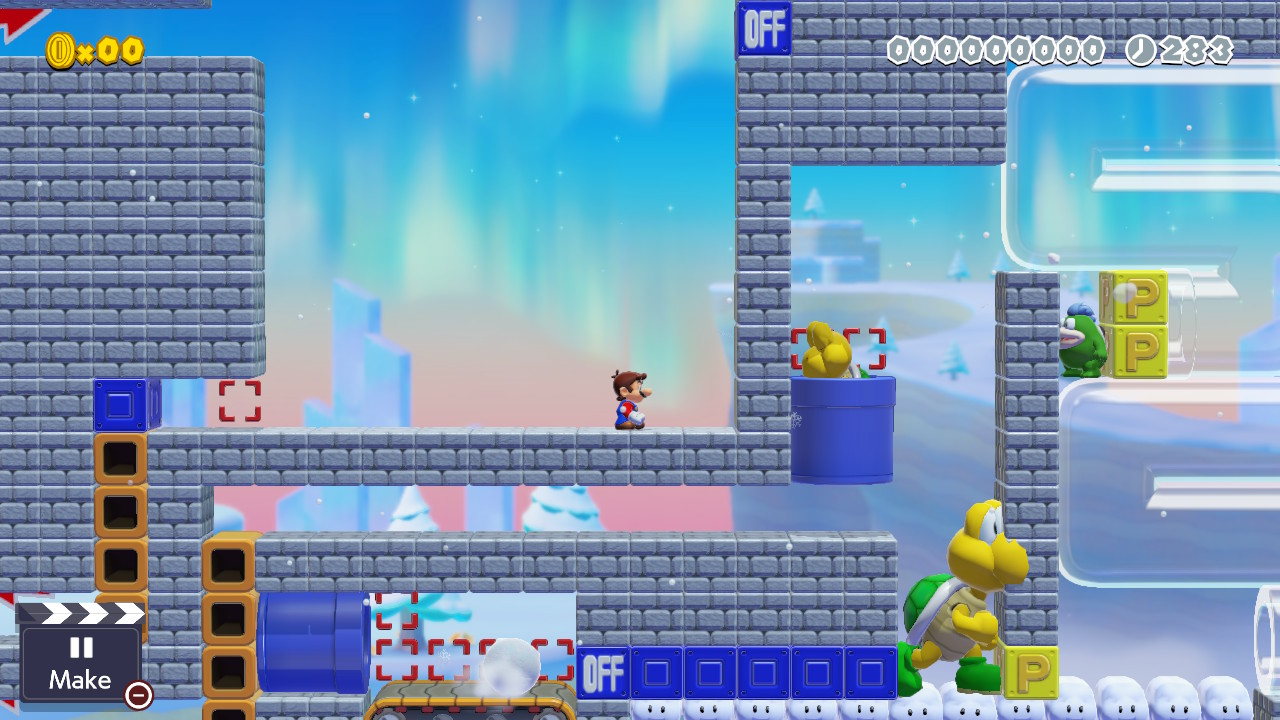}
        \caption{A snowball \ICON{smm3d/snowball}\ also spawns from the bottom blue pipe \ICON{smm3d/pipe}\ and rolls towards an on/off switch \ICON{smm3d/switch-off}.\newline\newline\newline\newline}
        \label{smm2-3dw-inc-3}
    \end{subfigure}\hfill
    \begin{subfigure}[t]{0.49\linewidth}
        \centering
        \includegraphics[width=\linewidth]{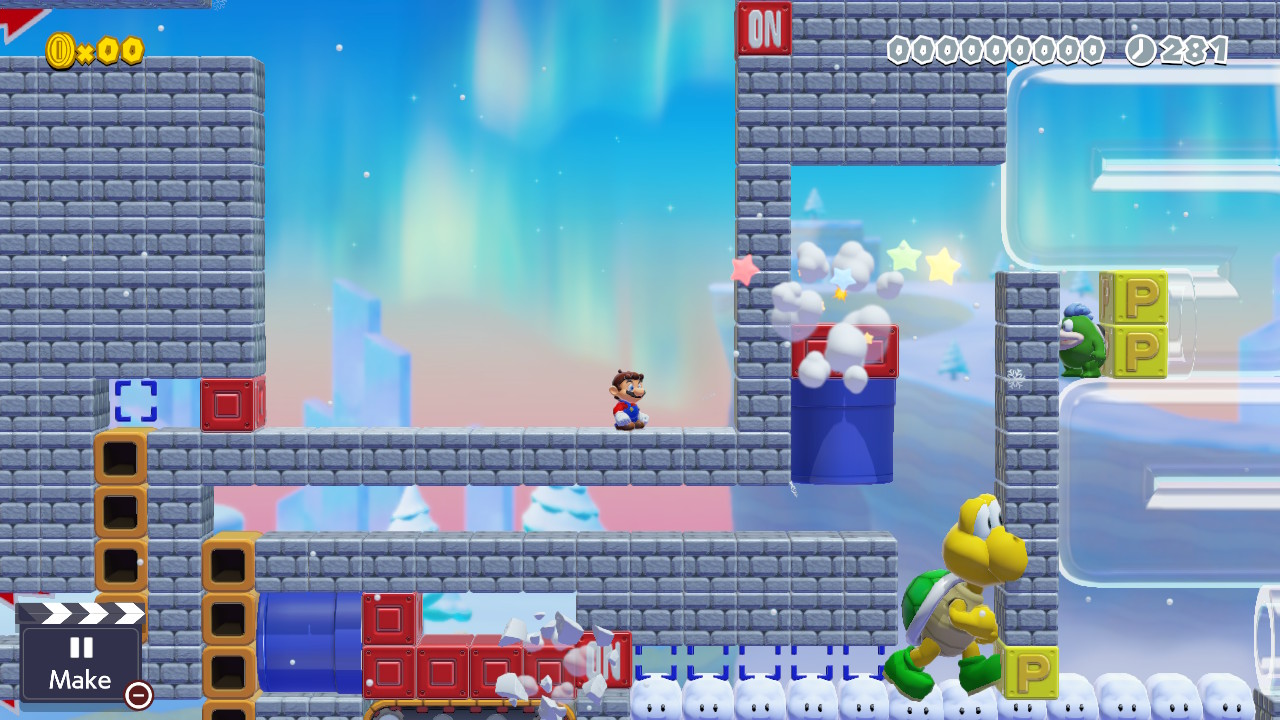}
        \caption{The snowball \ICON{smm3d/snowball}\ hits the on/off switch \ICON{smm3d/switch-on}, which blocks both blue pipes \ICON{smm3d/pipe}\ from spawning, allowing only one Koopa \ICON{smm3d/koopa}\ to be spawned and land on the global-ground cloud block \ICON{smm3d/cloud}. (Note that Mario should instead be to the left of the red block \ICON{smm3d/red-on}\ in a normal traversal; his current position is just to clearly see the gadget at work.)}
        \label{smm2-3dw-inc-4}
    \end{subfigure}

    \smallskip
    \begin{subfigure}{0.49\linewidth}
        \centering
        \includegraphics[width=\linewidth]{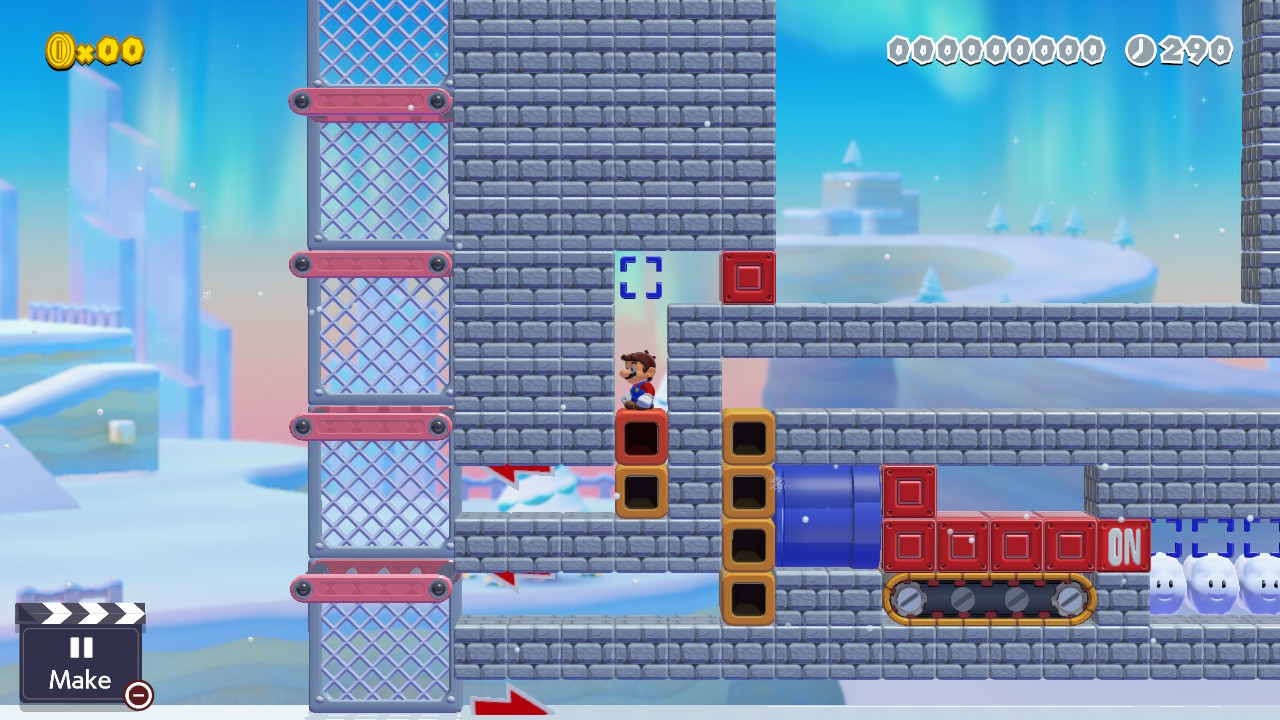}
        \caption{Mario exits through a stack of donut blocks \ICON{smm3d/donut}, which simulates a one-way.}
        \label{smm2-3dw-inc-5}
    \end{subfigure}
    \caption{Screenshots of traversing the Inc path for the 3D World game style of Super Mario Maker 2.}
    \label{smm2-3dw-inc}
    \end{figure}
\paragraph{Increment.}
    When Mario enters the increment path (whose entrance is the center-left tunnel in Figure \ref{smm2-3dw-inc-1} and whose exit is the center-left tunnel in Figure \ref{smm2-3dw-inc-5}), he must long jump to the right in order to hit an on/off switch \ICON{smm3d/switch-on}\ (as shown in Figure \ref{smm2-3dw-inc-1}). If he fails to do so, then he will be soft-locked when he falls into the chamber: it is wide enough that wall-jumping does not gain height, and tall enough that Mario cannot get back out with a twirl jump. 
    
    After Mario hits the switch, the pathway to the left is slightly opened, allowing the player to pass a deactivated red on/off block \ICON{smm3d/red-off}\ but still blocked from progressing by an activated blue on/off block \ICON{smm3d/blue-on}. At the same time, the deactivated on/off switch \ICON{smm3d/switch-off}\ means that the middle blue pipe \ICON{smm3d/pipe}\ is no longer blocked and begins spawning a Koopa \ICON{smm3d/koopa}\ (as shown in Figure \ref{smm2-3dw-inc-2}). Also, the left blue pipe is no longer blocked and spawns a snowball \ICON{smm3d/snowball}\ that rolls to the right against a conveyor \ICON{smm3d/conveyor}\ towards an on/off block \ICON{smm3d/switch-off}\ (as shown in Figure \ref{smm2-3dw-inc-3}). 
    
    The spawn rate of the pipes \ICON{smm3d/pipe}, the speed of the conveyor \ICON{smm3d/conveyor}, and the snowball's \ICON{smm3d/snowball}\ distance to the switch \ICON{smm3d/switch-off}\ are all timed so that there is enough time for exactly one Koopa \ICON{smm3d/koopa}\ to spawn before the switch is hit, blocking the blue pipes from spawning more Koopas or snowballs (as shown in Figure \ref{smm2-3dw-inc-4}) and allowing the player to continue through a drop of four donut blocks \ICON{smm3d/donut}\ that serves as a one-way to ensure that the gadget cannot be traversed in the wrong direction (as shown in Figure \ref{smm2-3dw-inc-5}).

    \begin{figure}
    \centering
    \begin{subfigure}[b]{0.45\linewidth}
        \centering
        \includegraphics[width=\linewidth]{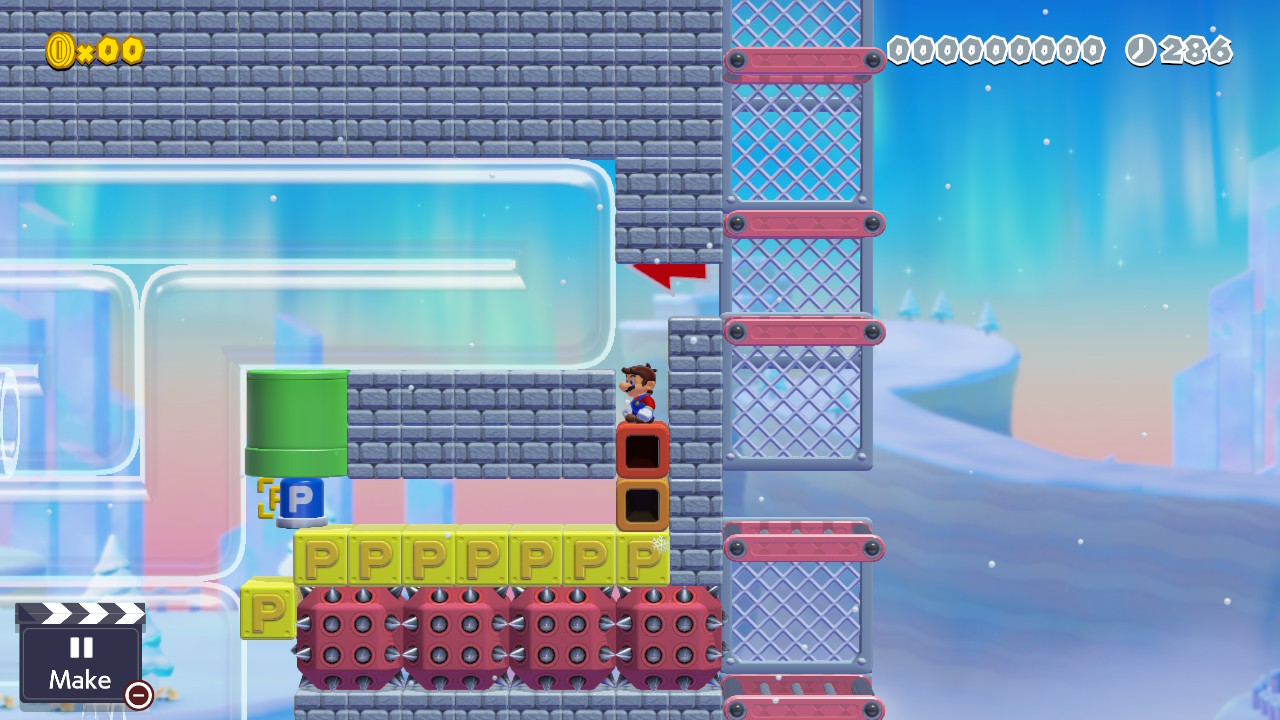}
        \caption{Mario enters through a stack of donut blocks \ICON{smm3d/donut}\ which simulate a one-way.}
        \label{smm2-3dw-dec-1}
    \end{subfigure}
    \begin{subfigure}[b]{0.45\linewidth}
        \centering
        \includegraphics[width=\linewidth]{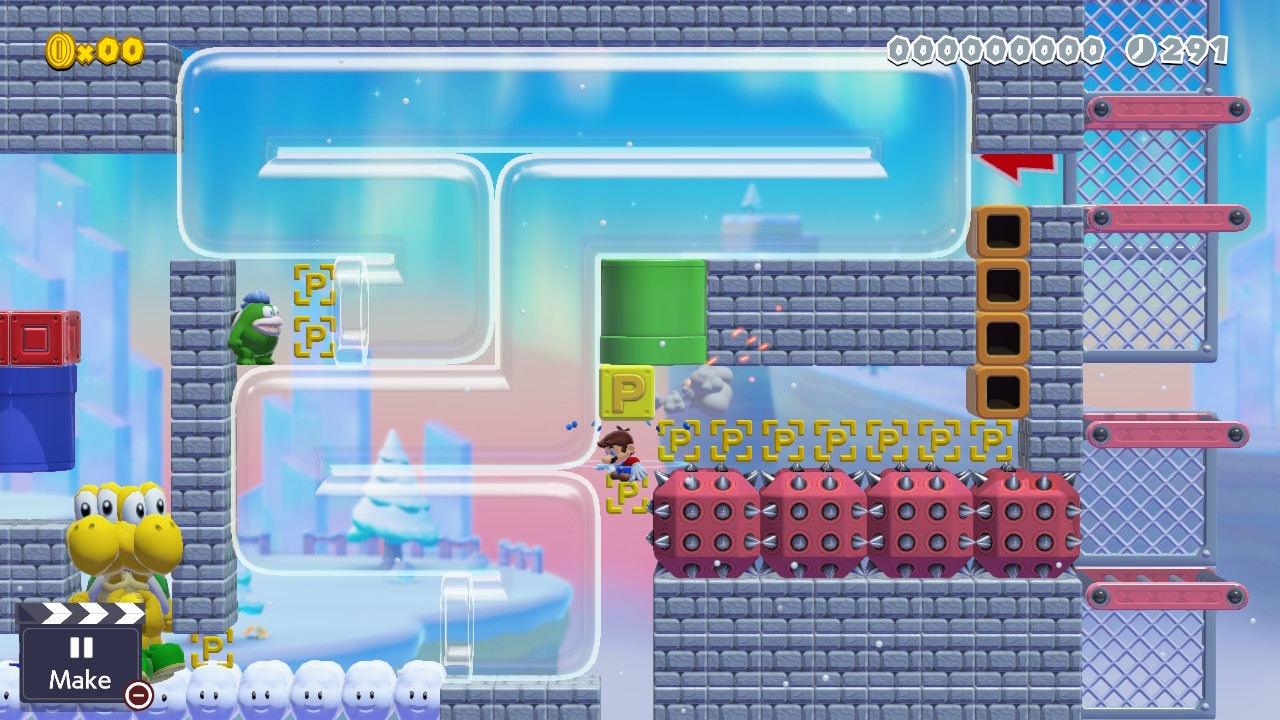}
        \caption{Mario hits the P-switch \ICON{smm3d/p-switch}, which opens the path for Spike \ICON{smm3d/spike}\ to throw a snowball \ICON{smm3d/snowball}\ at the Koopas \ICON{smm3d/koopa}.}
        \label{smm2-3dw-dec-2}
    \end{subfigure}

    \smallskip
    \begin{subfigure}[b]{0.45\linewidth}
        \centering
        \includegraphics[width=\linewidth]{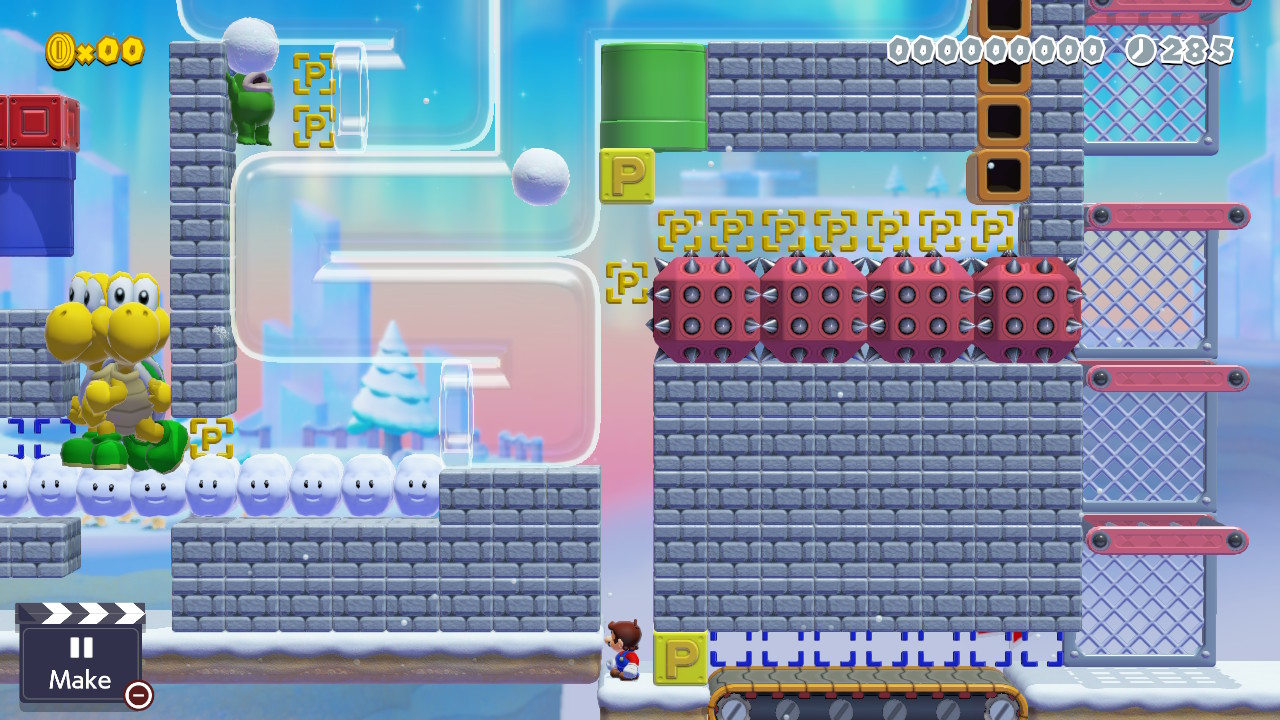}
        \caption{Spike \ICON{smm3d/spike}\ throws a snowball \ICON{smm3d/snowball}, which travels along the clear pipe \ICON{smm3d/clear}\ toward the Koopas \ICON{smm3d/koopa}.}
        \label{smm2-3dw-dec-3}
    \end{subfigure}
    
    \smallskip
    \begin{subfigure}[b]{0.45\linewidth}
        \centering
        \includegraphics[width=\linewidth]{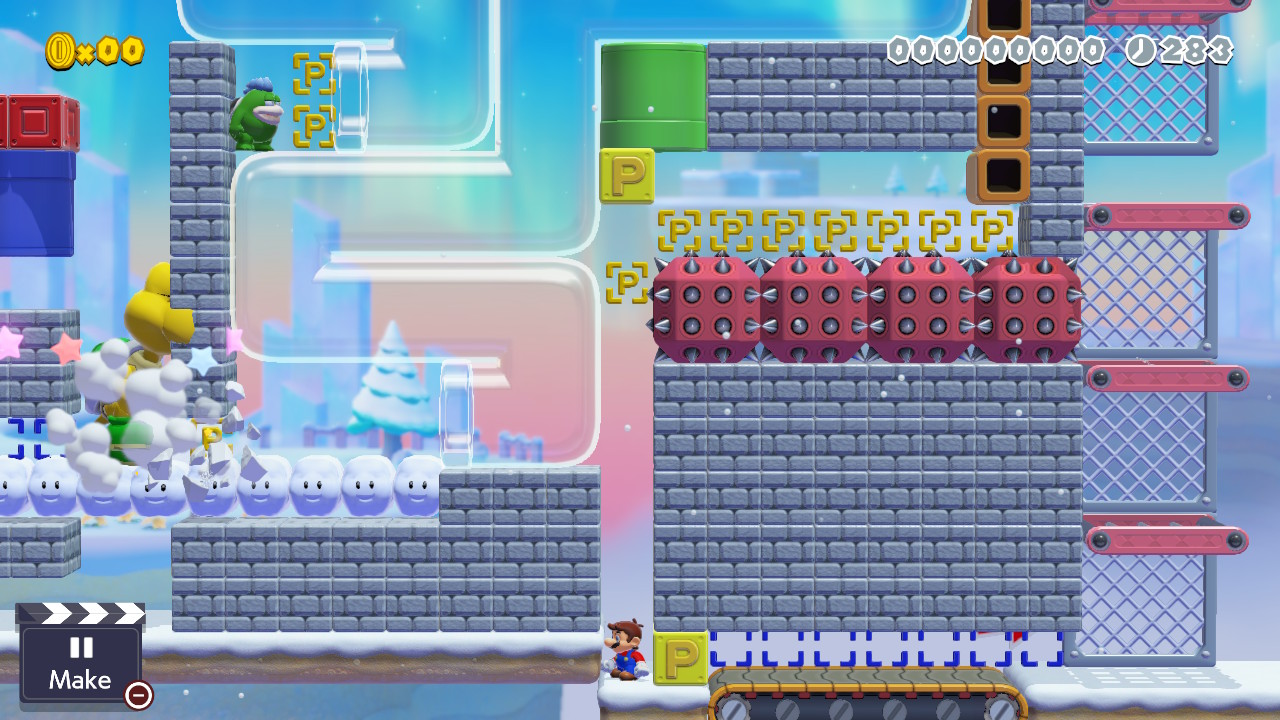}
        \caption{If there at least one Koopa \ICON{smm3d/koopa}, the snowball \ICON{smm3d/snowball}\ collides with it, killing exactly one Koopa and destroying the snowball.}
        \label{smm2-3dw-dec-4}
    \end{subfigure}
    \begin{subfigure}[b]{0.45\linewidth}
        \centering
        \includegraphics[width=\linewidth]{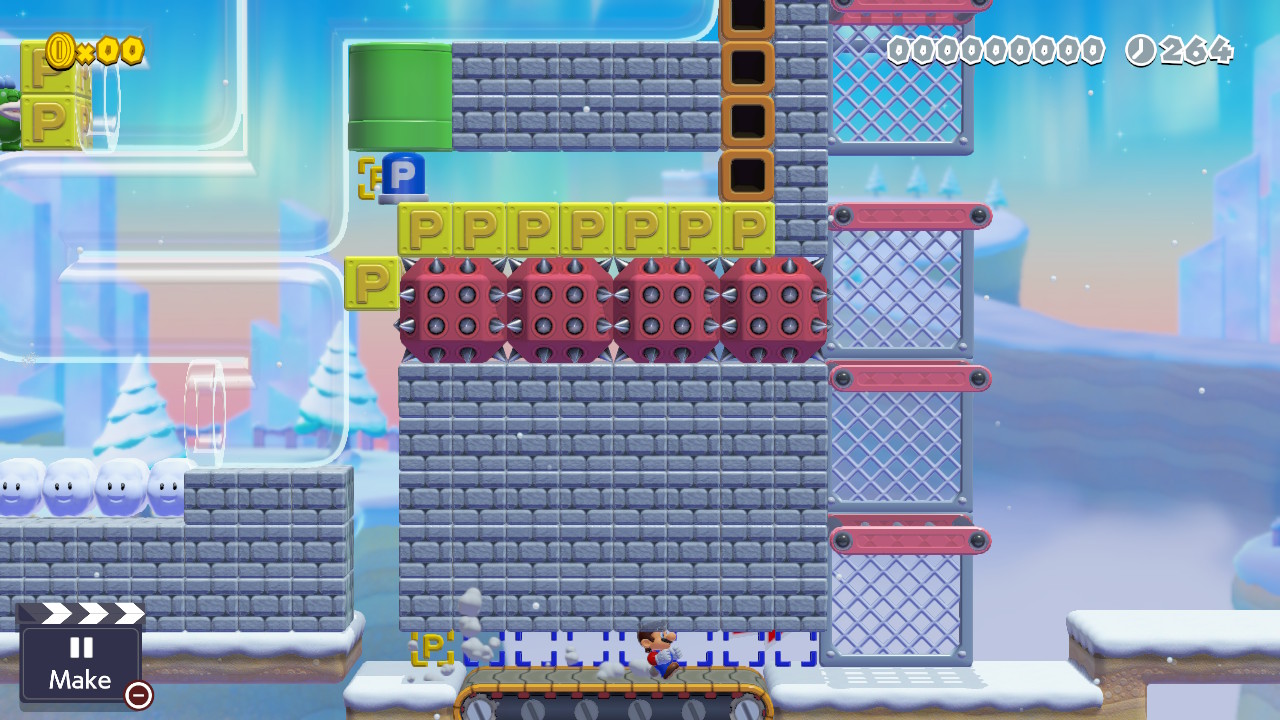}
        \caption{Then Mario exits after the P-switch \ICON{smm3d/p-switch}\ ends.\newline}
        \label{smm2-3dw-dec-5}
    \end{subfigure}
    
    \smallskip
    \begin{subfigure}[b]{0.45\linewidth}
        \centering
        \includegraphics[width=\linewidth]{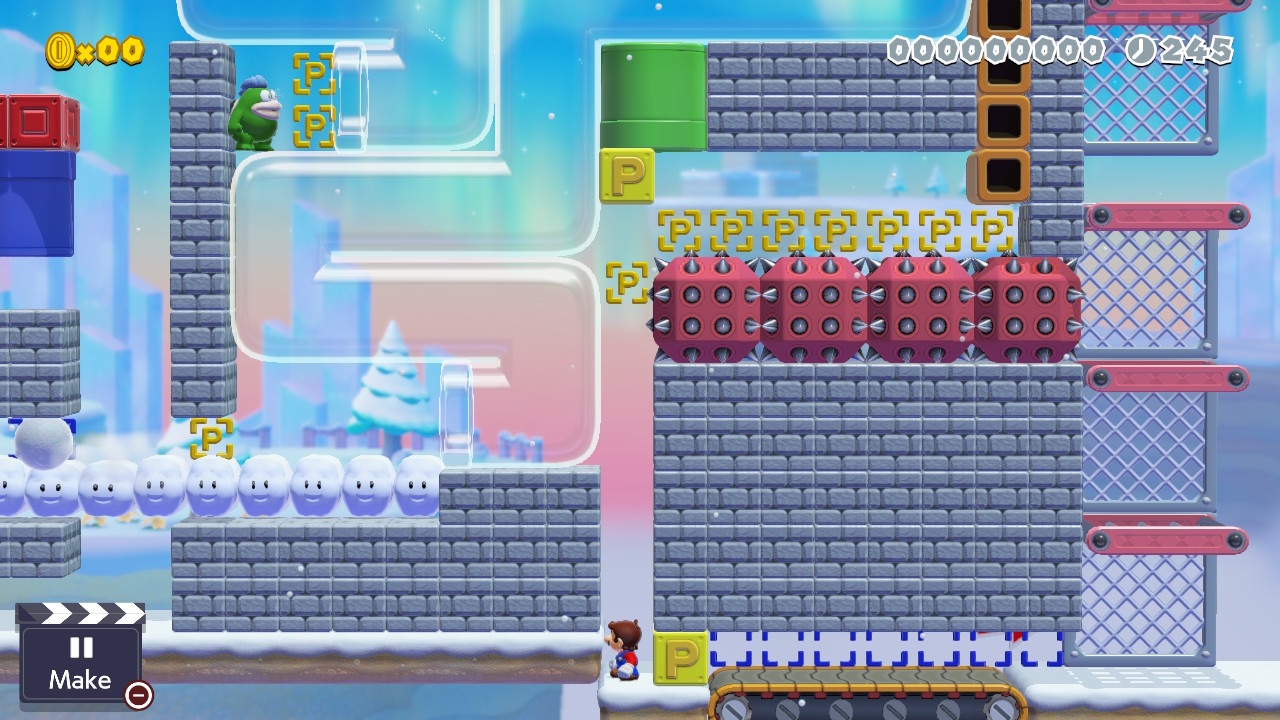}
        \caption{If there are no Koopas \ICON{smm3d/koopa}, the snowball \ICON{smm3d/snowball}\ rolls past, toward an on/off switch \ICON{smm3d/switch-on}\ offscreen to the left.}
        \label{smm2-3dw-bad-dec-1}
    \end{subfigure}
    \begin{subfigure}[b]{0.45\linewidth}
        \centering
        \includegraphics[width=\linewidth]{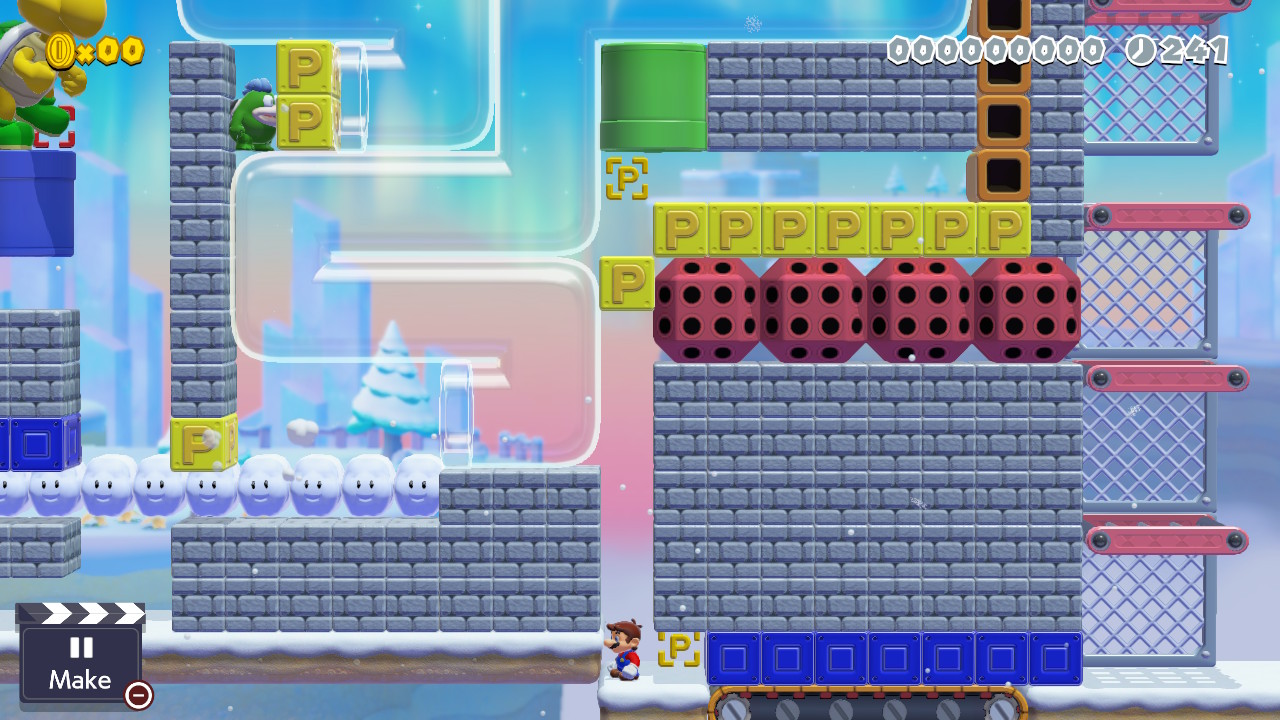}
        \caption{The on/off switch \ICON{smm3d/switch-off}\ is hit before the P-switch \ICON{smm3d/p-switch}\ ends, trapping Mario by the blue blocks \ICON{smm3d/blue-on}.}
        \label{smm2-3dw-bad-dec-2}
    \end{subfigure}
    \caption{Screenshots for traversing the DecNZ path for the 3D World game style of Super Mario Maker 2.}
    \label{smm2-3dw-dec}
    \end{figure}

\paragraph{Decrement.}
    When Mario enters the decrement path (whose entrance is the center-right tunnel in Figure \ref{smm2-3dw-dec-1} and whose exit is the bottom-right tunnel in Figure \ref{smm2-dec-5}), he once again goes through a drop of four donut blocks \ICON{smm3d/donut}\ that serves as a one-way to assure that he cannot take the P-switch \ICON{smm3d/p-switch}\ out of the decrement-path entrance (as shown in Figure \ref{smm2-3dw-dec-1}). 
    
    The tunnel has only one position (on the far left) that is two squares high, so this is the only place that Mario can activate the P-switch \ICON{smm3d/p-switch}\ (as shown in Figure \ref{smm2-3dw-dec-2}). The spike blocks \ICON{smm3d/spike-off}\ underneath the P-blocks \ICON{smm3d/p-block3d}\ ensure that Mario must fully move to the leftmost space when triggering the P-switch, to avoid taking damage \ICON{smm3d/spike-on}. The fact that the spike block is also along the side of the top of this drop chute ensures that Mario cannot try to break the gadget by wall-jumping up right as the P-switch timer ends. And the drop's height is large enough to not be undone by a twirl jump.
    
    Now that Mario is trapped in this chute and blocked from leaving by a P-block \ICON{smm3d/p-block3d}\ while the P-switch timer is running, Spike \ICON{smm3d/spike}\ throws a snowball \ICON{smm3d/snowball}\ into the clear pipe \ICON{smm3d/clear}\ (as shown in Figure \ref{smm2-3dw-dec-3}) and hits the mass of Koopas \ICON{smm3d/koopa}\ with the snowball, killing exactly one Koopa (as shown in Figure \ref{smm2-3dw-dec-4}). The length of the pipe and placement of P-blocks is carefully tuned so that exactly one snowball will get from Spike into the Koopa mass. 
    
    Lastly, if Mario tries to decrement in the state $=0$, then the snowball \ICON{smm3d/snowball}\ will continue rolling forward until it hits the on/off block \ICON{smm3d/switch-on}\ (as shown in Figure \ref{smm2-3dw-bad-dec-1}), breaking the counter but also soft-locking the player by activating the blue on/off blocks \ICON{smm3d/blue-on}\ to the right (as shown in Figure \ref{smm2-3dw-bad-dec-2}). The conveyor \ICON{smm3d/conveyor}\ prevents the player from moving too far away from the snowball before it is able to hit the on/off switch and trap Mario.
    %
    %
\end{proof}


\section{Open Problems}
\label{sec:open}

Of the 2D Mario Games released since New Super Mario Bros., we have shown that all except for Super Mario Wonder are undecidable, and a natural open question is whether it is too. There is evidence which suggests that it might be based on the presence of events and infinitely spawning Goombas, but the game is still very new, and more research is needed to understand the mechanics of the game well enough to make further claims about undecidability. 

There are also several older 2D Mario Games which have evidence that they might be able to build counters. In particular, any game with a Lakitu has a way of generating unlimited numbers of Spinies, Super Mario Bros.\ 3 and Super Mario World 2: Yoshi's Island both have enemies that can be generated from pipes, and Super Mario World has enemies generated by falling from the sky. Can we use any of the mechanics in those games to build counters? If any of these other games are not undecidable, it would also be noteworthy if we can obtain any other upper bounds on their complexity.

Finally, we showed that all games considered here are hard in constant-size levels, but it is not certain exactly what that constant is. We provided an explicit example of a single-screen counter in New Super Mario Bros.\ Wii, but we can definitely compact it further if we want to build the smallest possible universal counter machine. Furthermore, we can consider the Super Mario Maker games and whether it is possible to build a universal counter machine that fits inside of the standard constraint on level size.

\section*{Acknowledgments}

This paper was initiated during open problem solving in the MIT class
on Algorithmic Lower Bounds: Fun with Hardness Proofs (6.5440)
taught by Erik Demaine in Fall 2023.
We thank the other participants of that class for helpful discussions
and providing an inspiring atmosphere.

Portions of this paper originally appeared in Ani's master's thesis \cite{uuu}.

Several community level editors and emulators were very helpful in building and testing counters. In particular, we would like to thank:
\begin{itemize}
    \item \textit{Reggie!} by the NSMBW Community --- \url{https://github.com/NSMBW-Community/Reggie-Updated}
    \item \textit{CoinKiller} by Arisotura --- \url{https://github.com/Arisotura/CoinKiller}
    \item \textit{Miyamoto} by aboood40091 --- \url{https://github.com/aboood40091/Miyamoto}
    \item \textit{Dolphin} by the Dolphin Emulator Project --- \url{https://dolphin-emu.org/}
    \item \textit{Cemu} by Team Cemu --- \url{https://cemu.info/}
    \item \textit{Citra} by Citra Emu --- \url{https://github.com/citra-emu/citra}
\end{itemize}

\bibliographystyle{alpha}
\bibliography{citations}

\end{document}